\documentclass[letterpaper,11pt]{article}
\usepackage[utf8]{inputenc}

\usepackage[margin=1in]{geometry}
\usepackage{setspace}
\usepackage{graphicx}
\usepackage[dvipsnames,table]{xcolor}
\usepackage{fancyhdr} 
\usepackage{amsmath, amsfonts, amsthm, amssymb}
\usepackage{tikz}
\usepackage{authblk}%For affiliations
\usepackage{hyperref}
\usepackage{bbm}
\usepackage{dsfont}
\usepackage{subfig}
\usepackage{zref-clever}
\zcsetup{cap,noabbrev,nameinlink,sort}
\usepackage{booktabs}
\usepackage{microtype}  % fix bibliogrpahy spilling into margin
\usepackage{eurosym} %euro symbol
\usepackage{float}
\usepackage[normalem]{ulem} 
\usepackage{cancel}

\usepackage[title]{appendix}
\allowdisplaybreaks

\providecommand{\keywords}[1]
{
  {\small	
  \textbf{\textit{Keywords---}} #1}
}

\DeclareUnicodeCharacter{0301}{\'e}

\usepackage{xpatch}
\makeatletter
\AtBeginDocument{\xpatchcmd{\cref@thmoptarg}{\thm@headpunct{.}}{\thm@headpunct{}}{}{}}
\AtBeginDocument{\xpatchcmd{\cref@thmnoarg}{\thm@headpunct{.}}{\thm@headpunct{}}{}{}}
\makeatother

\newcommand{\cB}{\mathcal{B}}

\NewDocumentCommand{\newzctheorem}{mO{#1}m}{
  \newtheorem{#1}[sharedtheoremcounter]{#3}
    \AddToHook{env/#1/begin}{%
      \zcsetup{countertype={sharedtheoremcounter=#2}}}
}

\zcRefTypeSetup{assumption}{
  Name-sg = Assumption,
  name-sg = assumption,
  Name-pl = Assumptions,
  name-pl = assumptions,
}
\zcRefTypeSetup{optimization}{
  Name-sg = Optimization Problem,
  name-sg = optimization problem,
  Name-pl = Optimization Problems,
  name-pl = optimization problems,
}
\zcRefTypeSetup{claim}{
  Name-sg = Claim,
  name-sg = claim,
  Name-pl = Claims,
  name-pl = claims,
}

\newzctheorem{theorem}{Theorem}
\newzctheorem{proposition}{Proposition}
\newzctheorem{lemma}{Lemma}
\newzctheorem{corollary}{Corollary}
\newzctheorem{definition}{Definition}
\newzctheorem{example}{Example}
\newzctheorem{remark}{Remark}
\newzctheorem{assumption}{Assumption}
\newzctheorem{optimization}{Optimization Problem}
\newzctheorem{claim}{Claim}

\newcommand{\R}{{\mathds{R}}}
\newcommand{\E}{{\mathbb{E}}}
\newcommand{\Q}{{\mathbb{Q}}}
\renewcommand{\P}{{\mathbb{P}}}
\newcommand{\eps}{\varepsilon}

\newcommand{\bF}{{\mathbb{F}}}

\newcommand{\Rp}{{\R_+}}
\newcommand{\Rpnpo}{\R_+^{n+1}}

\newcommand{\cI}{{\mathcal{I}}}

\newcommand{\A}{{\mathcal{A}}}

\newcommand{\mfZ}{{\mathfrak{Z}}}
\newcommand{\mfB}{{\mathfrak{B}}}

\newcommand{\mfS}{{\mathfrak{S}}}

\newcommand{\F}{{\mathcal{F}}}

\newcommand{\N}{{\mathcal{N}}}

\renewcommand{\L}{{\mathcal{L}}}

\newcommand{\kinI}{k\in\cI}
\newcommand{\bz}{\boldsymbol{z}}
\newcommand{\bZ}{\boldsymbol{Z}}
\newcommand{\var}[1]{\mathrm{Var}^{#1}}

\newcommand{\newinf}{\mathop{\mathrm{inf}\vphantom{\mathrm{sup}}}}

\hypersetup{
  colorlinks   = true, %Colours links instead of boxes
  urlcolor     = blue, %Colour for external hyperlinks
  linkcolor    = red, %Colour of internal links
  citecolor   = blue %Colour of citations
}

\allowdisplaybreaks % fixes align environment weird spacing on page

\usepackage{todonotes}
\usepackage{setspace}

\usepackage{etoolbox}
\AtBeginEnvironment{align}{%
  \fontsize{10.25}{12.4}%
}
\AtBeginEnvironment{align*}{%
  \fontsize{10.25}{12.4}%
}
\AtBeginEnvironment{equation}{%
  \fontsize{10.25}{12.4}%
}
\AtBeginEnvironment{equation*}{%
  \fontsize{10.25}{12.4}%
}
\AtBeginEnvironment{alignat*}{%
  \fontsize{10.25}{12.4}%
}

\usepackage[natbib=true,
style=numeric,
maxcitenames=2,
giveninits]{biblatex}
\bibliography{references.bib}
\renewbibmacro{in:}{}  % suppresses In:

\title{Model Combination in Risk Sharing under Ambiguity}
\author[1]{Emma Kroell}
\author[2,a]{Sebastian Jaimungal}
\author[2,b]{Silvana M. Pesenti}

\affil[1]{Department of Mathematical Sciences, University of Copenhagen}
\affil[ ]{ek@math.ku.dk}
\affil[2]{Department of Statistical Sciences, University of Toronto}
\affil[a]{sebastian.jaimungal@utoronto.ca}
\affil[b]{silvana.pesenti@utoronto.ca}

\date{\today}

\begin{document}

\maketitle

\begin{abstract}
We consider the problem of an agent who faces losses in continuous time over a finite time horizon and may choose to share some of these losses with a counterparty. The agent is uncertain about the true loss distribution and has multiple models for the losses, characterized by a finite set of probability measures. Their goal is to optimize a mean-variance type criterion with model combination under ambiguity through risk sharing. We construct such a criterion using the chi-squared divergence, exploiting a dual representation to expand the state space, yielding a time consistent problem. Assuming a Cram\'er-Lundberg loss model, we fully characterize the optimal risk sharing contract and the agent’s wealth process under the optimal strategy. We prove that the strategy we obtain is admissible and that the value function satisfies the appropriate verification conditions. Furthermore, we show that the model combination problem is equivalent to the monotone mean-variance problem of Maccheroni et al.~(2009) under a composite probability measure that depends on the agent's reference probability measures. Finally, we apply the optimal strategy to an insurance setting in a simulation example and provide numerical illustrations of the results. 
\end{abstract}

\keywords{risk sharing, model combination, model ambiguity, monotone mean-variance, optimal contracting}

\onehalfspacing

\section{Introduction}

In many applications in insurance and finance, actors may have multiple, competing models for an uncertain outcome. Reasons for this could include the need to incorporate expert opinions or the availability of multiple data sources. For example, consider an insurance company with a model based on past claims, and who knows these claims will not be fully representative of future claims due to climate change. They could incorporate the threat of climate change by including additional models representing adverse climate change scenarios. Another example is a company moving into a new area of business, where they might need to include multiple sources of available data in their decision-making.

Our focus in this paper is on a model combination problem under ambiguity  inspired by this idea of different scenarios. We assume that an agent has access to multiple reference models, and does not know which of the reference models are correct. They must make a decision optimally accounting for ambiguity about these models, and therefore incorporate all models jointly into a single model penalization problem. In particular, we study the problem of an insurance company who faces insurance losses over a finite time horizon. The insurance company can share their risk with another agent, called the counterparty. The insurer then determines the optimal risk sharing strategy accounting for the multiple models and ambiguity about them. 

Traditionally, to incorporate model ambiguity in insurance and finance, one specifies a single reference model and searches for the worst-case outcome among alternate models, whose distance from the reference measure may be penalized using a divergence or distance measure. Commonly used divergence measures include the Kullback-Leibler (KL) divergence and the Wasserstein distance. Some seminal contributions among many include \textcite{HansenSargent2008} in economics, \textcite{Maenhout2004} and \textcite{PflugWozabal2007} in mathematical finance, \textcite{BlanchetMurthy2019} who quantify the divergence using optimal transport, and \textcite{Blanchet2019RobustActuarial} for applications in actuarial science. Given our inclusion of \emph{multiple} reference models, this work differs from these approaches to model uncertainty or robustness. 

Instead, it is more closely related to recent work on combining models in continuous time via barycenters. The barycenter of a set of probability measures $\P_k, k=1,\ldots, n$, has the form $\min_\Q \sum_{i=k}^n \pi_k D(\Q,\P_k)$ where $D$ is a divergence, $\pi_k$ are non-negative weights that sum to 1, and the minimum is taken over an appropriately defined set of probability measures. Work on this problem using the KL divergence includes  
\textcite{Liu2025}, who determine an efficient method for combining diffusion models by approximating the KL barycenter, and \textcite{JaimungalPesenti2026}, who determine the optimal model that combines differing expert models with constraints. 
In an optimal reinsurance context, \textcite{Kroell2025} study an optimal control and model combination problem under ambiguity and determine the optimal reinsurance premium charged by a reinsurer, showing that the optimal pricing model is a distortion of the barycenter model. Furthermore, a recent work using the adapted Wasserstein distance instead of the KL divergence is \textcite{acciaio2025multicausal}, who study a barycenter problem for stochastic processes with respect to causal optimal transport. 

Unlike the papers mentioned above, we study a model combination problem with the chi-squared divergence, which is defined as follows from a probability measure $\P$ to a probability measure $\Q$:
\begin{equation*}
    \chi^2(\Q\,\Vert \, \P) := \E^\P \left[ \left(\frac{d\Q}{d\P}\right)^2 - 1 \right] \,,
\end{equation*}
where $\frac{d\Q}{d\P}$ denotes the Radon-Nikodym derivative of $\Q$ with respect to $\P$. For an overview of probability metrics including the chi-squared divergence, see \textcite{gibbs2002}. 
Penalizing model ambiguity using chi-squared divergence is of particular interest due to its close relationship with another popular optimization criterion: the monotone mean-variance (MMV) preferences of \textcite{Maccheroni2006,Maccheroni2009}. MMV preferences were introduced as the minimal monotone extension of the mean-variance (MV) preferences. In particular, MMV preferences agree with mean-variance preferences when mean-variance preferences are monotone, and are closest pointwise to the mean-variance preferences when they are not. Given a reference measure $\P$, MMV preferences are defined as
\begin{equation}
\label{eqn:def_MMV}
    V^{\theta}_{\text{MMV}}(X;\P) :=
    \min_{\Q \in \Delta^2(\P)} \left( \E^\Q [X] + \frac{1}{2 \theta} \, \chi^2(\Q \, \Vert \, \P) \right) \,,
\end{equation}
where $X\in\L^2(\P)$, $\Delta^2(\P) = \{\Q \ll \P: \E^\P \left[\left(\frac{d\Q}{d\P}\right)^2\right] <\infty \}$, and $\theta >0$ is a parameter penalizing uncertainty. 

In our work, we incorporate a chi-squared divergence penalty between \emph{multiple} reference measures $\P_1, \ldots,\P_n$ and a candidate measure $\Q$. Similar to the barycenter problem, none of the reference measures is necessarily the ``true'' measure, and all are incorporated into the problem. This criterion incorporates both model ambiguity, through its construction as a penalized utility maximization problem, and model combination, through the weighted average term. Given a set of probability measures $\{\P_k\}_{k\in\N}$, where $\N := \left\{ 1, \ldots, n\right\}$, our model combination (MC) criterion is defined as:
\begin{equation}
    \label{eqn:def_MC}
    V^{\theta}_{\text{MC}}(X;\{\P_k\}_{k\in\N}) := \min_{\Q \in \Delta^2(\{\P_k\}_{k\in\N})} \left( \E^\Q [X] + \frac{1}{2 \theta} \sum_{k\in\N} \pi_k \, \chi^2(\Q \, \Vert \, \P_k) \right) \,,
\end{equation}
where $0 \leq \pi_k \leq 1$ are weights such that $\sum_{k\in\N} \pi_k = 1$ and $\Delta^2(\{\P_k\}_{k\in\N}):= \bigg\{\Q: \Q \ll \P_k \text{ and } \E^{\P_k} \left[\left(\tfrac{d\Q}{d\P_k}\right)^2\right] <\infty  \, \allowbreak  \text{ for all } k\in\N\bigg\}$. 
With this criterion, we generalize aforementioned works on model combination by using the chi-squared divergence, and works on MMV by incorporating multiple reference models $\P_k$.

Given the connection between our model combination under ambiguity problem and monotone mean-variance preferences, we mention here some important works that use MMV preferences. In recent years, there has been interest in applying MMV preferences to optimal investment problems. Examples include \textcite{TrybulaZawisza2019}, who study an investment problem in an incomplete market and \textcite{ShenZou2022} and \textcite{Hu2023} who study the problem with trading constraints. The general finding has been that the optimal strategies of MMV and classical mean-variance coincide in continuous time when there is only one reference model, and this is proved by \textcite{StrubLi2020} and \textcite{DuStrub2024} in general settings, under the assumption of continuous asset prices. \textcite{Cerny2020} gives an equivalent characterization of when MMV and MV preferences coincide in a general semimartingale model. Recently, \textcite{LiLiangPang2024,LiLiangPang2024arXiv} show that MMV and MV preferences coincide in a L\'evy market under a specific market assumption, and \textcite{Li2025ORL} expand this to a general market setting. There has also been recent interest in applying MMV to problems in optimal insurance; see \cite{LiGuo2021,LiGuo2024,ShiXu2024,Li2025CatastropheMMV} for examples in both the diffusive and jump process settings.

We also mention another related class of mean-variance problems, which are based on the concept of the Nash subgame-perfect equilibrium. These problems assure a time-consistent approach to the mean-variance problem by searching for an intrapersonal equilibrium point (see \textcite{TimeInconsistentControlBook2021} for an overview of this approach).  \textcite{ZengLiGu2016} use this perspective to find the robust reinsurance and investment strategies of an insurer. \textcite{Yuan2022} expand it to a leader-follower game between an insurance company and a reinsurance company, and find that model ambiguity increases the price of reinsurance. 
\textcite{ChenLandriaultLiLi2021risksharing} use this approach to study a risk sharing problem with $n$ insurers, and show that the Pareto optimal risk sharing strategy is a combination of a proportional and an excess-of-loss reinsurance strategy. Our work is distinct from these problems, as we assure a time consistent problem by expanding the state space rather than looking for equilibrium points.

This work has several contributions. Using the  criterion \eqref{eqn:def_MC}, we study a broad risk sharing problem, where the insurer may share any functional of its risk with a counterparty. We solve this problem and obtain closed-form expressions for the optimal strategies and the insurer's wealth process. The key to the approach is the introduction of auxiliary processes, which represent the Radon-Nikodym derivative between each reference model and the optimal model. Furthermore, we derive an explicit expression for the insurer's wealth process under the optimal strategy in terms of the auxiliary processes and find that the optimal wealth process is linear in the auxiliary processes. Using this result, we derive many properties of the optimal strategy, including its mean and variance. We find that the model penalization parameter penalizes the variance of the insurer's wealth process.

We further show that the model combination problem has an equivalent formulation as an MMV problem with respect to a composite probability measure, which is based on all the reference measures. We characterize this composite probability measure and show that, under it, the compensator of the Poisson random measure which drives the insurance losses is non-Markovian in general. Given the aforementioned results showing equivalence between MMV and mean-variance problems in continuous time, we also solve the mean-variance problem with respect to the composite measure, and show that it is indeed equivalent to both the MMV and original formulations of the problem.

The rest of the paper is structured as follows. \zcref{sec:setting-prelim} establishes the problem setting and requisite mathematical preliminaries. In \zcref{sec:main-results}, we state our optimization problem and solve for the insurer's optimal risk sharing strategy, as well as the optimal decision measure, $\Q^*$. We further develop semi-explicit formulas for the key processes, derive their properties under the $\Q^*$-measure, and provide a verification theorem. 
\zcref{sec:MMV-equiv} provides the equivalent formulation of the problem as an MMV problem with respect to a composite measure, and shows it also has an equivalent mean-variance formulation.
In \zcref{sec:one-model}, we apply our results to the simpler case where there is only one reference model, and in \zcref{sec:counterparty} we consider the pricing of such risk sharing contracts and determine how the counterparty could set the price. Finally, in \zcref{sec:data}, we apply our results to an example using data from a Spanish auto insurer, and illustrate the solution numerically.

\section{Problem setting}
\label{sec:setting-prelim}

Let $(\Omega, \F, \mathbb{F}=(\F_t)_{t \in [0,T]})$ be a completed and filtered measurable space. Let $\P_1, \ldots, \P_n, \P_C$ be $n+1$ equivalent probability measures defined on $(\Omega, \F, \mathbb{F})$ and let $\mathcal{I} := \{1,\ldots,n, C \}$, where the element $C$ denotes the index of the counterparty while the elements $1,\dots,n$ denote the indices of individual models.

Denote by $N(d\xi,dt)$ a Poisson random measure (PRM), which drives the insurance losses in the market. We assume that under a measure $\P_k$ for $k \in \cI$, $N$ has $\P_k$-compensator $\nu_k(d\xi,dt) = \nu_k(d\xi) dt$ and define the $\P_k$-compensated PRM by 
\begin{equation*}
    \tilde N^{\P_k}(d\xi,dt) = N(d\xi,dt) - \nu_k(d\xi)dt \,.
\end{equation*}
We assume that each compensator admits a density denoted by $v_k(\xi)$, i.e., $\nu_k(d\xi) = v_k(\xi) d\xi$ for $k \in \cI$. They have the same essential support on some subset of $\R_+$.

\begin{remark}
    \label{assump:equiv}
    As $\{\P_k\}_{k\in\cI}$ are a set of equivalent probability measures,  we have $\nu_j\ll\nu_k$, for all $j,k\in\cI$. That is, all compensators are absolutely continuous with respect to one another.     
\end{remark}
This implies that, when interpreting the random measure as a marked point process with the points denoting the time of occurrence of a loss and the mark its severity, the various models all agree on the support of the severity random variable.

Furthermore, we assume that for all $k \in \cI$
\begin{equation}
\label{eqn:nu_int_assumption}
    \int_{\Rp} \!\!\! \nu_k(d\xi) < \infty \quad \text{and} \quad \int_{\Rp} \!\!\! \xi  \, \nu_k(d\xi) < \infty \,,
\end{equation}
and that
\begin{equation}
    \int_{\Rp} \!\!\! \xi^2  \, \nu_C(d\xi) < \infty \,.
\end{equation} 
The first bound implies finite activity, the second and third imply  finite first and second moments of the severity distribution.

Next, we make two technical assumptions about the integrability of the compensators, which are required to assure the existence of the optimal risk sharing strategy. Essentially, these assumptions imply that the models for the loss distribution are not vastly different. 
Recall that the compensators have the same support.  Thus, when considering integrals over ratios of compensators (as in the following two assumptions), we use the convention that 0/0=0.

\begin{assumption}
\label{assump:intgr_comp}
For all $k \in \cI$
  \begin{equation*}
    \int_{\R_+} \frac{v_C^2(\xi)}{v_k(\xi)} \, d\xi  < \infty \,.
  \end{equation*}
\end{assumption}

This means that the parameters of the reference models must be within a certain range of the corresponding parameters of the counterparty's model.
For example, if we assume that $\nu_k(d\xi)$ is compound Poisson such that $v_k(\xi) = \lambda_k f_k(\xi)$, where $\lambda_k >0$ and $f_k$ is the density of a Gamma distribution with shape $m_k > 0$ and scale $\phi_k > 0$, \zcref{assump:intgr_comp} is satisfied if $2 m_C > m_k$ and $2 \phi_k > \phi_C$ for all $\kinI\setminus\{C\}$. Given that the differing models describe the same source of loss, this is a reasonable assumption.

\begin{assumption}
\label{assump:intgr_comp_jk}
For all $j,k \in \cI$
\begin{equation*}
  \int_{\Rp} \frac{v_C^3(\xi)}{v_j(\xi)v_k(\xi)} \, d\xi < \infty \,. 
\end{equation*}
\end{assumption}

Continuing the same example with a Gamma severity distribution, \zcref{assump:intgr_comp_jk} is satisfied if for all $j,k \in \cI\setminus\{C\}$, $3 m_C > m_j + m_k$ and $3 \phi_j\phi_k > \phi_C (\phi_j + \phi_k)$. Similarly to the interpretation of \zcref{assump:intgr_comp}, this means the reference models cannot be too different from the counterparty's model, and is a reasonable assumption for models that describe the same loss distribution.

\subsection{Insurer's wealth process}
\label{sec:21}

Our focus is on the behaviour of an insurance company who faces an insurable loss and receives premium income to compensate for this. The insurer charges a constant premium rate $c>0$. Initially, the insurer's wealth process follows the Cram\'er-Lundberg model:
\begin{equation}
\label{eqn:X_CL}
    X_t^{CL} = x_{0} + c \, t -  \int_0^t \!\! \int_{\Rp} \!\!\! \xi  \, N(d\xi,ds) \,,
\end{equation}
where $x_{0}>0$ is the insurer's initial wealth.

The insurer can share their risk with another agent, who we refer to as the counterparty. They could be another insurance company, a reinsurer, or another financial entity. The insurer shares a functional $\alpha_t(\xi)$ of the loss $\xi \in \Rp$ with the counterparty. The counterparty accepts the risk sharing $\alpha_t(\xi)$ and in turn charges a premium based on the amount of risk shared. We assume that the counterparty charges the expected value premium principle calculated under their own model, which is given by the probability measure $\P_C$. They markup the expected ceded loss by a markup rate $\eta>0$, also called the safety loading in an insurance setting. Thus the risk sharing premium is 
$p_C = (1+\eta) \int_{\R_+} \alpha_t(\xi) \nu_C (d\xi)$.\footnote{We could allow for time-dependency in the premia charged by the insurer and counterparty, i.e., $c(t)$ and $\eta(t)$, where $c(t)$ and $\eta(t)$ are bounded and integrable. Our results would hold but with additional notation to account for integration over time, so we prefer to set the premia constant.}
We assume that the risk sharing premium is set such that $c < (1+\eta) \int_{\R_+} \xi \, \nu_C (d\xi)$, i.e. passing the full risk to the counterparty is not optimal.

\begin{definition}[Admissible risk sharing strategies]
\label{def:alpha_int}
We define the set of admissible risk sharing strategies, $\A$, as those strategies $(\alpha_t(\cdot))_{t\in[0,T]}$ that are $\bF$-predictable random fields, $\alpha_t:\Rp\to\R$, satisfying
\begin{equation*}
    \E^{\P_C} \left[ \int_0^T \!\!\! \int_{\R_+} \!\!\! |\alpha_s(\xi)|^2 \,  \nu_C(d\xi) ds \right] < \infty \quad \text{and} \quad
    \E^{\P_C} \left[ \int_0^T \!\!\! \int_{\R_+} \!\! \left[ \xi - \alpha_s(\xi) \right]^2 \nu_C(d\xi) ds \right] < \infty \,.
\end{equation*}
\end{definition}

Given the above assumptions, the insurer's wealth process with risk sharing evolves according to the following equation:
\begin{equation}
\label{eqn:X_SDE}
    X_t^\alpha = x_{0} + \int_0^t \left[ c - (1+\eta) \int_{\Rp} \!\!\! \alpha_s(\xi) \, \nu_C (d\xi) \right] ds -  \int_0^t \!\! \int_{\Rp} \!\!\! [\xi - \alpha_s(\xi)] \, N(d\xi,ds) \,.
\end{equation}

\subsection{Insurer's criterion}
In this section, we introduce the insurer's criterion. The insurer is uncertain about the correct model for the insurable loss. They have access to $n+1$ models for the insurable loss, given by the probability measures $\P_k$ for $\kinI$, which correspond to the counterparty's model $\P_C$ and the $n$ other models $\P_1, \ldots, \P_n$. The insurer has varying degrees of certainty about the models, and may disregard some of them completely. Their goal is to determine the optimal way to share their risk with the counterparty, combining the models they have access to and taking into account uncertainty about the true loss distribution.

Our criterion is inspired by the monotone mean-variance criterion of \textcite{Maccheroni2006,Maccheroni2009}, given by \eqref{eqn:def_MMV}. 
To incorporate the $n+1$ reference models $\P_k$, $\kinI$, we propose the following criterion:
\begin{definition}
\label{def:MMVMA}
The agent's criterion for a given $\alpha\in\A$ is
\begin{equation}
\label{eqn:MMV1}
    V^{\theta}_{\text{MC}}(X^\alpha_T;\{\P_k\}_{\kinI}) := \min_{\Q \in \Delta^2 (\{\P_k\}_{k \in \cI})} \left\{ \E^\Q  \left[ X^\alpha_T \right]  + \frac{1}{2 \theta} \sum_{\kinI} \pi_k \, \E^{\P_k} \left[ \left(\frac{d\Q}{d\P_k}\right)^2 - 1 \right] \right\}  \,,
\end{equation}
where $0 \leq \pi_k \leq 1$ are weights such that $\sum_{\kinI} \pi_k = 1$ and 
\begin{equation*}
    \Delta^2 (\{\P_k\}_{k \in \cI}) := \left\{\Q: \Q \ll \P_k \text{ and } \E^{\P_k} \left[\left(\tfrac{d\Q}{d\P_k}\right)^2\right] <\infty  \, \text{ for all } k \in \cI \right\} \,.
\end{equation*}
\end{definition}
Note, $\P_C\in\Delta^2(\{\P_k\}_{k \in \cI})$, as $\P_C\ll \P_k$ for all $k\in\cI$, and $\E^{\P_k}\left[\left(\frac{d\P_C}{d\P_k}\right)^2\right]=e^{\int_{\Rp} \! \left(\frac{v_C(\xi)}{v_k(\xi)}-1\right)^2v_k(\xi)\,d\xi\,T}<\infty$ due to \zcref{assump:intgr_comp} and \eqref{eqn:nu_int_assumption}, hence $\Delta^2(\{\P_k\}_{k \in \cI})\neq\emptyset$.
Any of the models $\P_k$, $\kinI$, may be excluded from the insurer's criterion by setting $\pi_k=0$ for that model. 

It is convenient to express this criterion entirely under the measure $\Q$ by splitting the second term in \eqref{eqn:MMV1} and rewriting it as a $\Q$-expectation, giving
\begin{equation*}
    \min_{\Q \in \Delta^2(\{\P_k\}_{k \in \cI})} \E^\Q  \left[ X^\alpha_T + \frac{1}{2 \theta} \sum_{\kinI} \pi_k \left(\frac{d\Q}{d\P_k} - 1 \right)\right] \,.
\end{equation*}

Next, we obtain an alternate representation of \eqref{eqn:MMV1} by recasting the problem as a zero-sum stochastic game (see \textcite{mataramvuraoksendal2008}). To do so, we introduce a family of stochastic processes $\{Z^\beta_{k,t}\}_{\kinI,t\in[0,T]}$, which we call \textit{auxiliary processes}. These processes are driven by a random field $\beta$ which we call the \textit{compensator}. Each auxiliary process $Z_k^\beta$ may be viewed as the process version of the Radon-Nikodym derivative from $\P_k$ to a candidate probability measure $\Q_\beta$, which is parametrized by $\beta$, and reflects the effect of the reference model $\P_k$ on the insurer's problem. They are needed from a technical perspective to expand the state space of the problem, allowing for a time consistent solution.\footnote{By time consistent, we mean that Bellman's optimality principle may be applied to the problem (see \textcite{TimeInconsistentControlBook2021}).} We prove below that the set of probability measures induced by such processes generate the same set of measures $\Delta^2(\{\P_k\}_{k \in \cI})$ in the original problem.

For each $\bF$-predictable random field $(\beta_t(\cdot))_{t\in[0,T]}$, $\beta_t:\Rp\to\Rp$, we define the stochastic processes $\left\{\left(Z^\beta_{k,t}\right)_{k\in\cI}\right\}_{t\in[0,T]}$ as the stochastic exponentials that solve the SDEs:
\begin{equation}
\label{eqn:Z_SDE}
    dZ_{k,t}^\beta = - Z_{k,t^-}^\beta \int_{\Rp} \!\! \left[ 1 - \frac{\beta_t(\xi)}{v_k(\xi)}\right] \tilde N^{\P_k}(d\xi,dt)\,, \quad Z^\beta_{k,0} = 1 \,, \quad \text{for all } \kinI \,.
\end{equation}

\begin{definition}[Admissible compensators]
\label{def:beta_int2}
    Let $\mfB$ denote the $\bF$-predictable random fields $\beta$ such that, for all $\kinI$,
        \begin{equation}
        \label{eqn:finite-second-moment}
        \E^{\P_k} \left[ \left( Z_{k,T}^\beta \right)^2 \right] < \infty \,.
    \end{equation} 
\end{definition}
Note, for any $\beta\in\mfB$ and $k\in\cI$, as $Z_{k,T}^\beta$ is a non-negative random variable, we have $\E^{\P_k} \left[ Z_{k,T}^\beta \right]^2 \le \E^{\P_k} \left[ \left(Z_{k,T}^\beta\right)^2 \right]<\infty$. Hence, $Z_{k,T}^\beta\in\mathbb{L}^1$, and therefore a true martingale, and so $\E^{\P_k} \left[ Z_{k,T}^\beta \right]=Z_{k,0}^\beta=1$.

Fix $k\in\cI$, then for each $\beta\in\mfB$ we define the measure $\Q_{\beta,k}$ such that $\Q_{\beta,k}(\omega) = Z_{k,T}^\beta(\omega)\, \P_k(\omega)$ for all $\omega\in\F_T$.
From Girsanov's Theorem for random measures, under a candidate measure $\Q_{\beta,k}$,
\begin{align}
\label{eqn:girsanov}
    \tilde N^{\Q_{\beta,k}} (d\xi, dt) 
    &= \tilde N^{\P_k} (d\xi, dt) + \left[ 1 - \frac{\beta_t(\xi)}{v_k(\xi)}\right] \nu_k(d\xi)dt = N(d\xi, dt) -\beta_t(\xi) d\xi dt \,, 
\end{align}
is a martingale increment, i.e., $N$ has $\Q_{\beta,k}$-compensator $\beta_t(\xi) \, d\xi \, dt$. Note that as the $\Q_{\beta,k}$-compensator is independent of $k$, $\Q_{\beta,k}$ is independent of $k$, and we write this measure as $\Q_\beta$.

\begin{lemma}
    The equality of the set of probability measures $\Delta^2(\{\P_k\}_{k \in \cI}) = \{\Q_\beta: \beta\in\mfB\}$ holds.
\end{lemma}
\begin{proof}
We prove that any element of $\Delta^2(\{\P_k\}_{k \in \cI})$ has a corresponding element in $\mfB$ and vice versa.

Fix $k\in\cI$. By Theorem III.3.17, definitions III.3.15 and III.3.16, and equation II.8.2 of \cite{jacod2013limit}, for any $\Q\in\{\Q:\Q\ll\P_k\}$  the Radon-Nikodym derivative $\frac{d\Q}{d\P_k}=\mfZ^{k}_T$, where $(\mfZ_t^k)_{t\in[0,T]}$ satisfies the SDE $d\mfZ_t^k = \mfZ_{t^-}^k\int_{\Rp} (Y_{k,t}(\xi)-1)\,\tilde N^{\P_k} (d\xi, dt)$, $\mfZ_0^k=1$, for some random field $Y_k$ such that $Y_k\nu_k$ is $\sigma$-finite and $Y_k\nu_k$ is the $\Q$-compensator of the random measure $N$. 
As $\nu_k$ admits a density $v_k$, $Y_k\nu_k$ must also and its support must coincide with $\nu_k$. Denote the density corresponding to $Y_k\nu_k$ as $b_k$, then we have that $Y_k=\frac{b_k}{v_k}$. Further note, that the SDE admits the stochastic exponential form 
\begin{equation}
    \label{eqn:mfZ-stoch-exp}
    \mfZ_T^k=e^{-\int_0^T \! \int_{\Rp} \!\!  \left(Y_{k,t}(\xi)-1\right)v_k\,d\xi\,dt + \int_0^T \! \int_{\Rp} \!\! \log Y_{k,t}(\xi)\,N(d\xi,dt)}.
\end{equation}

As $\P_l\ll \P_k$ for any $l\in\cI$ and $k$ still fixed as above, we similarly have that $\frac{d\P_l}{d\P_k}=\mfZ^{k,l}_T$, where $(\mfZ_t^{k,l})_{t\in[0,T]}$ satisfies the SDE $d\mfZ_t^{k,l} = \mfZ_{t^-}^{k,l}\int_{\Rp} (Y_{k,l,\, t}(\xi)-1)\,\tilde N^{\P_k} (d\xi, dt)$, $\mfZ_0^{k,l}=1$, and $Y_{k,l}=\frac{v_l}{v_k}$.

Next, take $\Q\in\Delta^2(\{\P_k\}_{k \in \cI})$, we then have
\begin{align}
    \mfZ_T^k=\frac{d\Q}{d\P_k}=\frac{d\Q}{d\P_l}\frac{d\P_l}{d\P_k}&=\mfZ_{T}^{l}\,\mfZ_{T}^{k,l}
    \\
    \begin{split}
        &= e^{-\int_0^T \! \int_{\Rp} \!\!  \left(\frac{b_l(t,\xi)}{v_l(t,\xi)}-1\right)v_l\,d\xi\,dt + \int_0^T \! \int_{\Rp} \!\! \log\frac{b_l(t,\xi)}{v_l(t,\xi)}\,N(d\xi,dt)}
        \\
        &\quad \times 
        e^{-\int_0^T \! \int_{\Rp} \!\! \left(\frac{v_l(t,\xi)}{v_k(t,\xi)}-1\right)v_k\,d\xi\,dt + \int_0^T \! \int_{\Rp} \!\! \log\frac{v_l(t,\xi)}{v_k(t,\xi)}\,N(d\xi,dt)}
    \end{split}
    \\
    &=
        e^{-\int_0^T \! \int_{\Rp} \!\! \left(\frac{b_l(t,\xi)}{v_k(t,\xi)}-1\right)v_k\,d\xi\,dt + \int_0^T \! \int_{\Rp} \!\! \log\frac{b_l(t,\xi)}{v_k(t,\xi)}\,N(d\xi,dt)} \,.\label{eqn:mfZ-stoch-exp-2}
\end{align}
Therefore, by identifying \eqref{eqn:mfZ-stoch-exp} and \eqref{eqn:mfZ-stoch-exp-2}, we have $Y_k=\frac{b_l}{v_k}$. From earlier, we have that $Y_k=\frac{b_k}{v_k}$, hence $b_l=b_k$ for all $k,l\in\cI$. We denote this simply by $b$. Furthermore, note that $\mfZ_T^k=Z_{k,T}^b$ where $Z^b_k$ is given by \eqref{eqn:Z_SDE}.

Therefore, we have shown that for any $\Q\in\Delta^2(\{\P_k\}_{k \in \cI})$, there exists a single random field $b$, such that $\frac{d\Q}{d\P_k}=Z_{k,T}^b$ for all $k\in\cI$. Moreover, for any $\Q\in\Delta^2(\{\P_k\}_{k \in \cI})$, we have $\E^{\P_k} \left[\left(Z_{k,T}^b\right)^2\right]= \E^{\P_k} \left[\left(\tfrac{d\Q}{d\P_k}\right)^2\right] <\infty$ (by the definition of the set $\Delta^2(\{\P_k\}_{k \in \cI})$), hence, $b\in\mfB$.

We next show the other direction. Choose $\beta\in\mfB$, define the Radon-Nikodym derivatives $\frac{d\Q_\beta}{d\P_k}=Z_{k,T}^\beta$, with $\{(Z_{k,t}^\beta)_{k\in\cI}\}_{t\in[0,T]}$ satisfying the SDEs \eqref{eqn:Z_SDE} (note as argued in \eqref{eqn:girsanov}, the probability measure $\Q_\beta$ does not depend on $k$). The induced probability measure $\Q_\beta$ satisfies (i) $\Q_\beta\ll\P_k$, for all $k\in\cI$ and (ii) $\E^{\P_k}\left[\left(\frac{d\Q_\beta}{d\P_k}\right)^2\right]=\E^{\P_k}\left[(Z_{k,T}^\beta)^2\right] <\infty$ (as $\beta\in\mfB$).  Hence, the corresponding $\Q_\beta\in\Delta^2(\{\P_k\}_{k \in \cI})$.
\end{proof}

From this lemma, we immediately obtain the following corollary.
\begin{corollary}
The criterion \eqref{eqn:MMV1} is equivalent to
\begin{equation*}
    \inf_{\beta \in \mfB} \E^{\Q_\beta}  \left[ X^\alpha_T + \frac{1}{2 \theta} \sum_{\kinI} \pi_k \left(Z^\beta_{k,T} - 1 \right)\right] \,.
\end{equation*}
\end{corollary}

\section{The insurer's optimal risk sharing strategy under model combination}
\label{sec:main-results}

In this section, we first state the insurer's optimization problem and solve it for the insurer's optimal risk sharing strategy, denoted $\alpha^*$, and the optimal probability measure, denoted $\Q^* := \Q_{\beta^*}$. We then show that under the optimal risk sharing strategy, the insurer's wealth $X$ and the auxiliary processes $Z_k$, $\kinI$, may be written in a semi-explicit form, and use this to determine the mean and variance of the processes at any time in $[0,T]$. In particular, we show that there is a linear relationship between the optimal $X$ and the $Z_k$'s. Finally, we provide a verification theorem, showing the candidate controls are indeed optimal.

\subsection{Optimization problem and candidate solution}
\begin{optimization}
\label{opt:insurer-problem}
    The insurer seeks the solution to the following problem:
    \begin{equation*}
        \sup_{\alpha \in \A} \newinf_{\beta \in \mfB} \E^{\Q_\beta}  \left[ X^\alpha_T + \frac{1}{2 \theta} \sum_{\kinI} \pi_k \left(Z^\beta_{k,T} - 1 \right)\right] \,,
    \end{equation*}
    where
    \begin{subequations}
         \label{eqn:SDEs}
        \begin{align}
        \label{eqn:X_Q_SDE}
            d X_t^\alpha &= \left[ c - (1+\eta) \int_{\Rp} \!\!\! \alpha_t(\xi) \, \nu_C (d\xi)  - \int_{\Rp} \!\!\! [\xi - \alpha_t(\xi)] \,\beta_t(\xi) \, d\xi  \right] dt - \int_{\Rp} \!\!\! [\xi - \alpha_t(\xi)] \, \tilde N^{\Q_\beta}(d\xi,dt) \,,\\
            dZ_{k,t}^\beta &=  Z_{k,t^-}^\beta \int_{\Rp} \!\! \left[ 1 - \frac{\beta_t(\xi)}{v_k(\xi)}\right]^2 \nu_k(d\xi)dt - Z_{k,t^-}^\beta \int_{\Rp} \!\! \left[ 1 - \frac{\beta_t(\xi)}{v_k(\xi)}\right] \tilde N^{\Q_\beta}(d\xi,dt)  \,, \qquad \forall\;\kinI,
        \end{align}
    \end{subequations}
    with $X^\alpha_0 = x_{0}$, $Z^\beta_{k,0}=1$ for all $\kinI$.
\end{optimization}
We obtain the dynamics of $X^\alpha$ and $Z_k^\beta$, $\kinI$, given in \eqref{eqn:SDEs} by writing the original dynamics given in \eqref{eqn:X_SDE} and \eqref{eqn:Z_SDE} in terms of the candidate measure $\Q_\beta$.

Let $\bZ_t = (Z_{1,t}, \ldots, Z_{n,t}, Z_{C,t})$ be the vector version of the auxiliary processes and let $\bz = (z_1, \ldots, z_n, z_C)$ be an arbitrary vector in $\Rpnpo$. For fixed controls $\alpha \in \A$, $\beta \in \mfB$, we define the time-$t$ version of the value function as 
\begin{equation}
    J^{\alpha,\beta}(t,x,\bz) :=  \E^{\Q_\beta}_{t,x,\bz}  \left[ X^\alpha_T + \frac{1}{2 \theta} \sum_{\kinI} \pi_k \left(Z^\beta_{k,T} - 1 \right)\right] \,,
\end{equation}
where $\E^{\Q_\beta}_{t,x,\bz}[\cdot]$ denotes the $\Q_\beta$-expectation given that the processes $X$, $\bZ$ at time $t^-$ are equal to $x$ and $\bz$, respectively, i.e., $X_{t^-}=x$ and $\bZ_{t^-} = \bz$. We then define the insurer's time-$t$ optimal value function at the optimal controls as
\begin{equation*}
    J(t,x,\bz) := \sup_{\alpha \in \A} \newinf_{\beta \in \mfB} J^{\alpha,\beta}(t,x,\bz) \,.
\end{equation*}

We first derive a candidate value function for \zcref{opt:insurer-problem} and the associated candidate controls. Let $\mfS$ denote the set of functions $g:[0,T] \times \Rp \times \R \times \Rpnpo \to \R$ such that $\int_{\Rp} g(t,\xi, x, \bz) \,d\xi < \infty$ for every $t,x,\bz \in [0,T] \times \R \times \Rpnpo$.
The generator of the stochastic differential equations \eqref{eqn:SDEs} for Markov controls $a, b \in \mfS$ is
\begin{align}
\label{eqn:generator}
    \begin{split}
        A^{a,b} &f(t,x,\bz) 
        \\
        :=&\, \partial_t f(t,x,\bz) 
        \\
        & + \left[ c - (1+\eta) \int_{\Rp} \!\!\! a(t,\xi,x,\bz) \, \nu_C (d\xi)  - \int_{\Rp} \!\!\! [\xi - a(t,\xi,x,\bz)] \,b(t,\xi,x,\bz) \, d\xi  \right] \partial_x f(t,x,\bz)\\
        & + \sum_{\kinI} z_{k} \int_{\Rp} \!\! \left[ 1 - \frac{b(t,\xi,x,\bz)}{v_k(\xi)}\right]^2 \!\! \nu_k(d\xi) \,\partial_{z_k} f(t,x,\bz) 
        \\
	    & + \int_{\Rp} \Bigg[f \! \left(t,x - [\xi -  a(t,\xi,x,\bz)],z_1 \frac{b(t,\xi,x,\bz)}{v_1(\xi)},\ldots, z_n \frac{b(t,\xi,x,\bz)}{v_n(\xi)},z_C \frac{b(t,\xi,x,\bz)}{v_C(\xi)} \right) - f(t,x,\bz)\\
        & \quad\qquad \, + [\xi -  a(t,\xi,x,\bz)] \partial_x f(t,x,\bz)+  \sum_{k\in\cI} z_k \left( 1 - \frac{b(t,\xi,x,\bz)}{v_k(\xi)}\right) \partial_{z_k} f(t,x,\bz)\Bigg] b(t,\xi,x,\bz) d\xi \,,
    \end{split}  
\end{align}
where $f$ is a continuously differentiable function on $(0,T) \times \R \times \Rpnpo$.

The next result states the candidate controls and value function for \zcref{opt:insurer-problem}. We provide a verification result in \zcref{sec:verification}. 
\begin{proposition}
\label{prop:optim_contr}
    The candidate controls for \zcref{opt:insurer-problem} in feedback form are 
    \begin{subequations}
        \begin{align}
            \alpha^*(t,\xi,\bz) &= \xi  - \frac{1}{\theta} \sum_{\kinI} \pi_k \, z_k \, \ell_k(T-t) \left[ (1+\eta) \frac{v_C(\xi)}{v_k(\xi)} - 1 \right] \,,
            \label{eqn:alpha-star}
            \\
            \beta^*(\xi) &= (1+\eta) \, v_C(\xi) \,,
            \label{eqn:beta-star}
        \end{align}
    \end{subequations}
        where
    \begin{equation}
    \label{eqn:ell_k}
        \ell_k(t) = \exp \left( t \int_{\Rp} \!\! \left[1- (1+\eta) \frac{v_C(\xi)}{v_k(\xi)} \right]^2\!\! \nu_k(d\xi)\right) \,,
    \end{equation}
    and the candidate value function is
    \begin{equation*}
        J(t,x,\bz) = x + \frac{1}{2 \theta} \sum_{\kinI} \pi_k \, z_k \, \ell_k(T-t) - \frac{1}{2\theta} - \left[ (1+\eta) \int_{\Rp} \!\! \xi \, \nu_C(d\xi) - c \right] (T-t) \,.
    \end{equation*}
\end{proposition}
Before providing the proof, we first comment on the candidate optimal controls and value function.
First, note that by \zcref{assump:intgr_comp}, the integral over $\xi$ in the exponential in $\ell_k(t)$ is finite and thus $\ell_k(t) < \infty$ for all $t \in [0,T]$. Therefore, the candidate value function is well-defined. The form of the compensator of the optimal measure, $\beta^*$, is instructive: we see that the probability measure the insurer uses to determine the optimal risk sharing agreement is tied to the premium charged for the risk sharing. They use the counterparty's model for the losses, but increase the rate of arrival multiplicatively by $1+\eta$, where $\eta$ is the counterparty's markup rate to the premium. The form of the optimal risk sharing strategy, $\alpha^*$, is more complex. One may view it as a generalization of an excess-of-loss reinsurance contract, which is usually of the form $(\xi-d)_+$ for some retention limit $d>0$. In this case, however, the ``retention limit'' depends on the size of the loss, $\xi$, and there is no restriction that it be positive.

\begin{proof}
The insurer's value function
\begin{equation*}
    J(t,x,\bz) = \sup_{\alpha \in \A} \newinf_{\beta \in \mfB}\E^{\Q_\beta}_{t,x,\bz}  \left[ X^\alpha_T + \frac{1}{2 \theta} \sum_{\kinI} \pi_k \left(Z^\beta_{k,T} - 1 \right)\right] \,,
\end{equation*}
must satisfy the Hamilton-Jacobi-Bellman-Isaacs (HJBI) equation
\begin{align}
     0 & = \,\sup_{a \in \mfS} \newinf_{b \in \mfS} \big\{ A^{a,b} J(t,x,\bz) \big\} \label{eqn:HJBI1} \\
    & = \, \partial_t J(t,x,\bz) \nonumber \\
    & \quad + \sup_{a \in \mfS} \newinf_{b \in \mfS} \!\! \Bigg\{ \left[ c - (1+\eta) \int_{\Rp} \!\!\! a(t,\xi,x,\bz) \, \nu_C (d\xi)  - \int_{\Rp} \!\!\! [\xi - a(t,\xi,x,\bz)] \,b(t,\xi,x,\bz) \, d\xi  \right] \partial_x J(t,x,\bz) \nonumber \\
    & \hspace{6.5em} + \sum_{\kinI} z_{k} \int_{\Rp} \!\! \left[ 1 - \frac{b(t,\xi,x,\bz)}{v_k(\xi)}\right]^2 \!\! \nu_k(d\xi) \,\partial_{z_k} J(t,x,\bz) \nonumber \\
	& \hspace{6.5em} + \int_{\Rp} \Big[ J \! \left(t,x - [\xi -  a(t,\xi,x,z)],z_1 \frac{b(t,\xi,x,z)}{v_1(\xi)},\ldots,z_n \frac{b(t,\xi,x,z)}{v_n(\xi)},z_C \frac{b(t,\xi,x,z)}{v_C(\xi)}\right) \nonumber \\
    &  \hspace{9.5em} - J(t,x,\bz) + [\xi -  a(t,\xi,x,\bz)] \partial_x J(t,x,\bz) \nonumber \\
    &  \hspace{9.5em} +  \sum_{k\in\cI} z_k \left[ 1 - \frac{b(t,\xi,x,\bz)}{v_k(\xi)}\right] \partial_{z_k} J(t,x,\bz)\Big] b(t,\xi,x,\bz) d\xi \Bigg\} \nonumber 
\end{align}
with terminal condition
\begin{equation}
\label{eqn:HJBI2}
J(T,x,\bz) = x + \frac{1}{2\theta} \sum_{\kinI} \pi_k(z_k -1) \,.
\end{equation}
Using the Ansatz
\begin{equation*}
    J(t,x,\bz)= x + \sum_{\kinI} A_k(t)z_k + B(t) \,,
\end{equation*}
where $A_k$, $\kinI$, and $B$ are deterministic functions satisfying $A_k(T) = \pi_k/2\theta,\, \kinI \,, B(T) = -1/2\theta$, 
we obtain the simplified HJBI equation:
\begin{equation}
\begin{split}
    0 &=  \sum_{\kinI} A'_k(t) \, z_k  + \,  B'(t) \\
    & \quad + \sup_{a \in \mfS} \newinf_{b \in \mfS} \!\Bigg\{ \! c - (1+\eta) \int_{\Rp} \!\!\! a(t,\xi,x,\bz) \, \nu_C (d\xi)  - \int_{\Rp} \!\!\! [\xi - a(t,\xi,x,\bz)] \,b(t,\xi,x,\bz) \, d\xi  \\
    & \hspace{6em} + \sum_{\kinI} z_{k} A_k(t) \int_{\Rp} \!\! \left[ 1 - \frac{b(t,\xi,x,\bz)}{v_k(\xi)}\right]^2 \!\! \nu_k(d\xi) \Bigg\} \,.    
\end{split}
\label{eqn:HBJI-anstaz}
\end{equation}
Next, consider the infimum problem for $b$, where we minimize the following functional:
\begin{equation*}
    \mathcal{L}_1[b] := \int_{\Rp} \!\! \left( \, \sum_{\kinI} z_{k} A_k(t) \!\! \left[ 1 - \frac{b(t,\xi,x,\bz)}{v_k(\xi)}\right]^2 \!\! v_k(\xi) - [\xi - a(t,\xi,x,\bz)] \,b(t,\xi,x,\bz) \right) d\xi \,.
\end{equation*}
Let $\eps >0$ and $f$ be an arbitrary function such that $b + \eps f \in \mfS$. We apply a variational first order condition to obtain the equation: 
\begin{equation*}
    0 = \lim_{\eps \searrow 0} \frac{1}{\eps} \left( \mathcal{L}[b + \eps f] - \mathcal{L}[b] \right) = \int_{\Rp} \!\!\! f(t, \xi, x, \bz) \left( \, 2 \sum_{\kinI} z_{k} A_k(t) \!\! \left[\frac{b(t,\xi,x,\bz)}{v_k(\xi)} - 1\right] - [\xi - a(t,\xi,x,\bz)] \right) d\xi \,.
\end{equation*}
As the function $f$ is arbitrary, the above expression vanishes for
\begin{equation}
    \label{eqn:b_feedback}
    \hat b(t,\xi,x,\bz) = \frac{\xi - a(t,\xi,x,\bz) + 2 \sum_{\kinI} z_{k} A_k(t) }{ 2 \sum_{\kinI} A_k(t) \frac{z_{k} }{v_k(\xi)}} \,.
\end{equation}
Substituting this form back into the  HJBI equation \eqref{eqn:HBJI-anstaz} yields the new equation
\begin{equation}
\label{eqn:HBJI-b-subbed-in}
    \begin{split}
     0 = &\, \sum_{\kinI} A'_k(t) \, z_k  + B'(t) + c + \sum_{\kinI} A_k(t) \, z_k \int_{\Rp} \!\!\! \xi \, \nu_k (d\xi) \\
    &+ \sup_{a \in \mfS} \left\{ - \int_{\Rp} \!\! \left[ (1+\eta) \, a(t,\xi,x,\bz) \, v_C (\xi)  + \frac{\left(\xi - a(t,\xi,x,\bz) + 2 \sum_{\kinI} z_{k} A_k(t) \right)^2}{ 4\sum_{\kinI} A_k(t) \frac{z_{k} }{v_k(\xi)}} \right] d\xi \right\} \,, 
\end{split}
\end{equation}
with terminal conditions $A_k(T) = \pi_k/2\theta,\, \kinI \,, B(T) = -1/2\theta$.
Next, we solve the supremum problem for $a$ using a variational first order condition. To this end, we minimize the functional 
\begin{equation*}
    \mathcal{L}_2[a] :=  \int_{\Rp} \!\! \left[ (1+\eta) \, a(t,\xi,x,\bz) \, v_C (\xi)  + \frac{\left(\xi - a(t,\xi,x,\bz) + 2 \sum_{\kinI} z_{k} A_k(t) \right)^2}{ 4\sum_{\kinI} A_k(t) \frac{z_{k} }{v_k(\xi)}} \right] d\xi\,.
\end{equation*}
Let $\eps >0$ and $g$ be an arbitrary function such that $b + \eps g \in \mfS$. Then 
\begin{equation*}
    0 = \lim_{\eps \searrow 0} \frac{1}{\eps} \left( \mathcal{L}_2[a + \eps g] - \mathcal{L}_2[a] \right) = \int_{\Rp} \!\!\! g(t, \xi, x, \bz) \left(  (1+\eta) \,  v_C (\xi)  - \frac{\xi - a(t,\xi,x,\bz) + 2 \sum_{\kinI} z_{k} A_k(t) }{2\sum_{\kinI} A_k(t) \frac{z_{k} }{v_k(\xi)}} \right) d\xi \,.
\end{equation*}
As $g$ is arbitrary, we obtain the following feedback form for $a$:
\begin{equation}
    \label{eqn:a_feedback}
    \hat a(t,\xi,x,\bz) =\xi  - 2 \sum_{\kinI} A_k(t) \, z_k \left[ (1+\eta) \frac{v_C(\xi)}{v_k(\xi)} - 1 \right] \,.
\end{equation}
Substituting this expression into the HBJI equation \eqref{eqn:HBJI-b-subbed-in}, we obtain
\begin{align*}
     0 = &\, \sum_{\kinI} A'_k(t) \, z_k  + B'(t) + c - (1+\eta) \int_{\Rp} \!\!\! \xi \, \nu_C(d\xi) + \sum_{\kinI} A_k(t) \, z_k \int_{\Rp} \!\! \left[ 1 - (1+\eta) \frac{v_C(\xi)}{v_k(\xi)}\right]^2 \nu_k(d\xi) \,
\end{align*}
with terminal conditions $A_k(T) = \pi_k/2\theta,\, \kinI \,, B(T) = -1/2\theta$. Solving for the unknown functions $A_k(t)$, $\kinI$, and $B(t)$ gives
\begin{align*}
    A_k(t) &= \frac{\pi_k}{2\theta} \exp \left( \int_{\Rp} \left[1- (1+\eta) \frac{v_C(\xi)}{v_k(\xi)} \right]^2\!\! \nu_k(d\xi) \, (T-t) \right) , \quad \kinI \,, \\
    B(t) &=  \left[ (1+\eta) \int_{\Rp} \!\! \xi \, \nu_C(d\xi) - c \right] (t-T) - \frac{1}{2\theta} \,.
\end{align*}
Substituting these functions into the Ansatz gives the insurer's value function:
    \begin{align*}
        J(t,x,\bz)=& \, x + \frac{1}{2 \theta} \sum_{\kinI} \pi_k \, z_k \, \ell_k(T-t) - \frac{1}{2\theta} - \left[ (1+\eta) \int_{\Rp} \!\! \xi \, \nu_C(d\xi) - c \right] (T-t) \,,
    \end{align*}
    where 
    \begin{equation*}
         \ell_k(t) := \exp \left( t \int_{\Rp} \left[1- (1+\eta) \frac{v_C(\xi)}{v_k(\xi)} \right]^2\!\! \nu_k(d\xi)  \right) \,.
    \end{equation*}
    Furthermore, simplifying \eqref{eqn:b_feedback} and \eqref{eqn:a_feedback} and substituting in the $A_k(t)$, $\kinI$, gives the candidate optimal controls.
\end{proof}

\subsection{Processes under optimal controls}
\label{subsec:processes}
Next, we derive expressions for the processes $X$ and $\bZ$ under the candidate controls $\alpha^*$ and $\beta^*$, which we denote by $X^*$ and $\bZ^*$, respectively. We find that $X^*$ and $\bZ^*$ have a relatively simple form, which allows for explicit calculation of their mean and variance. The processes under the optimal controls satisfy the SDEs:
\begin{equation*}
    \begin{split}
         d X_t^* &=\left[c - (1+\eta) \int_{\Rp} \!\!\! \xi \, \nu_C(d\xi) + \frac{1+\eta}{\theta} \sum_{\kinI} \pi_k Z^*_{k,t^-} \ell_k(T-t) \int_{\Rp} \!\!\! \left[ (1+\eta) \frac{v_C(\xi)}{v_k(\xi)} - 1 \right] \nu_C(d\xi)  \right] dt \\
         & \quad - \frac{1}{\theta} \sum_{\kinI} \pi_k Z^*_{k,t^-} \ell_k(T-t) \int_{\Rp} \!\!\! \left[ (1+\eta) \frac{v_C(\xi)}{v_k(\xi)} - 1 \right]  N(d\xi,dt) \,, \\
         X_0^* &= x_{0} \,,
    \end{split}
\end{equation*}
where
    \begin{equation*}
        \ell_k(t) = \exp \left(t \int_{\Rp} \!\! \left[1- (1+\eta) \frac{v_C(\xi)}{v_k(\xi)} \right]^2\!\! \nu_k(d\xi) \right) \quad \text{for } \kinI \,,
    \end{equation*}
and for $\kinI$,
\begin{equation*}
    \begin{split}
         dZ^*_{k,t} &= Z^*_{k,t^-} \int_{\Rp} \!\!\!  \left[ v_k(\xi) - (1+\eta) v_C(\xi) \right] d\xi dt -  Z^*_{k,t^-} \int_{\Rp} \!\!\!  \left[ 1 - (1+\eta) \frac{v_C(\xi)}{v_k(\xi)} \right]  N(d\xi,dt) \,, \\
         Z^*_{k,0} &= 1 \,. 
    \end{split}
\end{equation*}

Using It\^o's lemma, we solve these SDEs and find that $X^*$ may be written as a linear combination of the $Z^*_k$'s.
\begin{proposition}
\label{prop:expression-Z-X}
For $t \in [0,T]$:
    \begin{align}
        Z^*_{k,t} &= \exp \left(t \int_{\Rp} \!\!\!  \left[ v_k(\xi) - (1+\eta) v_C(\xi) \right] d\xi + \int_0^t \!\! \int_{\Rp} \!\!\!  \ln \left( (1+\eta) \frac{v_C(\xi)}{v_k(\xi)} \right)  N(d\xi,ds) \right)\,, \quad \kinI \,, \nonumber \\
        X^*_t &= x_{0} + \left[ c - (1+\eta) \int_{\Rp} \!\!\! \xi \, \nu_C(d\xi) \right] t + \frac{1}{\theta} \sum_{\kinI} \pi_k \, \ell_k(T) \, \left[1 - \ell_k(-t) Z^*_{k,t}\right] \,.\nonumber
    \end{align}
\end{proposition}

\begin{proof}
    For $Z^*_{k,t}$, applying It\^{o}'s formula for semimartingales to the function $\ln(Z^*_{k,t})$ for each $\kinI$ gives 
    \begin{align*}
        \ln(Z^*_{k,t}) = & \, \ln(Z^*_{k,0}) + \int_0^t \frac{1}{Z^*_{k,s^-}} dZ^*_{k,s} \\
        & +\int_0^t \!\! \int_{\Rp} \!\!\! \left\{ \ln\left(Z^*_{k,s^-} - Z^*_{k,s^-} \left[ 1 - (1+\eta) \frac{v_C(\xi)}{v_k(\xi)} \right] \right) - \ln(Z^*_{k,s^-}) +  \left[ 1 - (1+\eta) \frac{v_C(\xi)}{v_k(\xi)} \right]\right\} N(d\xi,ds) \\
        = & \int_0^t \!\! \int_{\Rp} \!\!\!  \left[ v_k(\xi) - (1+\eta) v_C(\xi) \right] d\xi ds - \int_0^t \!\!\int_{\Rp} \!\!\!  \left[ 1 - (1+\eta) \frac{v_C(\xi)}{v_k(\xi)} \right]  N(d\xi,ds) \\
        & +\int_0^t \!\! \int_{\Rp} \!\!\! \left\{ \ln\left( (1+\eta) \frac{v_C(\xi)}{v_k(\xi)} \right) + \left[ 1 - (1+\eta) \frac{v_C(\xi)}{v_k(\xi)} \right]\right\} N(d\xi,ds) \\
        = & \int_0^t \!\! \int_{\Rp} \!\!\!  \left[ v_k(\xi) - (1+\eta) v_C(\xi) \right] d\xi ds +\int_0^t \!\! \int_{\Rp} \!\!\!  \ln\left( (1+\eta) \frac{v_C(\xi)}{v_k(\xi)} \right) N(d\xi,ds) \,. 
    \end{align*}
    Exponentiation gives the result. \\
    For $X^*_t$, define the function 
    \begin{equation*}
        f(t, X^*_t, \bZ^*_t) = X^*_t + \frac{1}{\theta} \sum_{\kinI} \pi_k \,  \ell_k(T-t) Z^*_{k,t}  - \left[ c - (1+\eta) \int_{\Rp} \!\!\! \xi \, \nu_C(d\xi) \right] t - \frac{1}{\theta} \sum_{\kinI} \pi_k \, \ell_k(T) \,.
    \end{equation*}
    Then by  It\^{o}'s formula for multi-dimensional semimartingales (see, e.g., \textcite[Chapter 2, Theorem 33]{Protter2005}), we have
    \begin{align*}
        f(t, X^*_t, \bZ^*_t) = & \, f(0, X^*_0,  \bZ^*_0) + \int_0^t \frac{\partial f}{\partial s} (s, X^*_{s^-}, \bZ^*_{s^-}) \, ds \\
        &+ \int_0^t \frac{\partial f}{\partial x} (s, X^*_{s^-}, \bZ^*_{s^-}) dX^*_s + \sum_{\kinI} \int_0^t \frac{\partial f}{\partial z_k} (s, X^*_{s^-}, \bZ^*_{s^-}) dZ^*_{k,s} \\
        &+ \int_0^t \!\! \int_{\Rp} \!\!\! \Big\{ f \left(s, X^*_{s^-} + \Delta X^*_{s}, \bZ^*_{s^-} + \Delta \bZ^*_{s} \right) - f(s, X^*_{s^-},  \bZ^*_{s^-})  \\
        & \qquad \qquad \quad  - \frac{\partial f}{\partial x} (s, X^*_{s^-},  \bZ^*_{s^-}) \Delta  X^*_{s}  - \sum_{\kinI} \frac{\partial f}{\partial z_k} (s, X^*_{s^-},  \bZ^*_{s^-}) \Delta Z^*_{k,s} \Big\} N(d\xi,ds) \,,
    \end{align*}
    where for a process $Y$, $\Delta Y_s = Y_s - Y_{s^-}$ is the jump at $s$. In particular, $\Delta X^*_s = - \frac{1}{\theta} \sum_{\kinI} \pi_k Z^*_{k,s^-} \ell_k(T-s) \left[ (1+\eta) \frac{v_C(\xi)}{v_k(\xi)} - 1 \right]$ is the jump in $X^*$ at $s$ and $\Delta \bZ^*_s := (\Delta Z^*_{1,s} , \ldots, \Delta Z^*_{n,s}, \Delta Z^*_{C,s} ) = \Big(-Z^*_{1, s^-} \left[1 - (1+\eta) \frac{v_C(\xi)}{v_1(\xi)} \right],$ $\ldots, -Z^*_{n, s^-} \left[1 - (1+\eta) \frac{v_C(\xi)}{v_n(\xi)} \right],\eta Z^*_{C, s^-}\Big)$ is the jump in the vector $\bZ^*$ at $s$. We have that
    \begin{align*}
        0 &= f \left(s, X^*_{s^-} + \Delta X^*_{s}, \bZ^*_{s^-} + \Delta \bZ^*_{s} \right) - f(s, X^*_{s^-},  \bZ^*_{s^-}) \,\, \text{and } \\
        0 &= - \frac{\partial f}{\partial x} (s, X^*_{s^-},  \bZ^*_{s^-}) \Delta  X^*_{s}  - \sum_{\kinI} \frac{\partial f}{\partial z_k} (s, X^*_{s^-},  \bZ^*_{s^-}) \Delta  Z^*_{k,s} \,.
    \end{align*}
    Furthermore, by the definition of $\ell_k(t)$, we have that for $\kinI$,
    \begin{align*}
       \partial_t \, \ell_k(T-t) &= - \exp \left((T-t) \int_{\Rp} \!\! \left[1- (1+\eta) \frac{v_C(\xi)}{v_k(\xi)} \right]^2\!\! \nu_k(d\xi)  \right) \int_{\Rp} \!\! \left[1- (1+\eta) \frac{v_C(\xi)}{v_k(\xi)} \right]^2\!\! \nu_k(d\xi)   \\
       &= - \ell_k(T-t) \int_{\Rp} \!\! \left[1- (1+\eta) \frac{v_C(\xi)}{v_k(\xi)} \right]^2\!\! \nu_k(d\xi)   \,.
    \end{align*}

    Substituting in the derivatives and jumps, we obtain,
    \begin{align*}
        f(t, X^*_t, \bZ^*_t) = x_{0} &+ \int_0^t \left[ - \frac{1}{\theta} \sum_{\kinI} \pi_k \,  \ell_k(T-s) Z^*_{k,s} \int_{\Rp} \!\! \left[1- (1+\eta) \frac{v_C(\xi)}{v_k(\xi)} \right]^2\!\! \nu_k(d\xi) - c + (1+\eta) \int_{\Rp} \!\!\! \xi \, \nu_C(d\xi) \right] ds \\
        &+ \int_0^t \left[c - (1+\eta) \int_{\Rp} \!\!\! \xi \, \nu_C(d\xi) + \frac{1+\eta}{\theta} \sum_{\kinI} \pi_k Z^*_{k,s^-} \ell_k(T-s) \int_{\Rp} \!\! \left[ (1+\eta) \frac{v_C(\xi)}{v_k(\xi)} - 1 \right] \nu_C(d\xi)  \right] ds \\
        & - \int_0^t \frac{1}{\theta} \sum_{\kinI} \pi_k \, Z^*_{k,s^-} \ell_k(T-s) \int_{\Rp} \!\! \left[ (1+\eta) \frac{v_C(\xi)}{v_k(\xi)} - 1 \right]  N(d\xi,ds)   \\
        & + \sum_{\kinI} \int_0^t \frac{\pi_k}{\theta} \, \ell_k(T-s) Z^*_{k,s^-}  \int_{\Rp} \!\!\!  \left[ v_k(\xi) - (1+\eta) v_C(\xi) \right] d\xi ds \\
        &- \sum_{\kinI} \int_0^t \frac{\pi_k}{\theta} \, \ell_k(T-s)  Z^*_{k,s^-} \int_{\Rp} \!\!  \left[ 1 - (1+\eta) \frac{v_C(\xi)}{v_k(\xi)} \right]  N(d\xi,ds) 
    \end{align*}
    and upon cancellation, we have
    \begin{equation*}
        f(t, X^*_t, \bZ^*_t)= x_{0} \,.
    \end{equation*}
    Substituting in the definition of $f(t, X^*_t, \bZ^*_t)$ gives the result.
\end{proof}

Next, we use \zcref{prop:expression-Z-X} to calculate the expected value of the auxiliary processes $Z_k^*$, $\kinI$, under the reference measures $\P_k$ and the optimal measure $\Q^*$.

\begin{corollary}
\label{cor:EZ}
    For $t\in[0,T]$ and all $\kinI$, $\E^{\P_k}[Z^*_{k,t}] = 1$  and $ \E^{\Q^*}[Z^*_{k,t}] = \ell_k(t)$.
\end{corollary}
\begin{proof}
    Under each $\P_k$, $\kinI$, the PRM $N$ has compensator $\nu_k(d\xi) \,dt = v_k(\xi)d\xi \,dt$. Using the representation of $Z_{k}^*$ from \zcref{prop:expression-Z-X} and the exponential formula for Poisson random measures, we compute
    \begin{align*}
        \E^{\P_k}[Z^*_{k,t}] &= \exp \left(t \int_{\Rp} \!\!\!  \left[ v_k(\xi) - (1+\eta) v_C(\xi) \right] d\xi \right) \E^{\P_k}\!\left[  \exp \left(\int_0^t \!\! \int_{\Rp} \!\!\!  \ln \left( (1+\eta) \frac{v_C(\xi)}{v_k(\xi)} \right)  N(d\xi,ds) \right)\right] \\
        &=  \exp \left( t \int_{\Rp} \!\!\!  \left[ v_k(\xi) - (1+\eta) v_C(\xi) \right] d\xi \right) \exp \left(t \int_{\Rp} \!\!\! \left[ (1+\eta) \frac{v_C(\xi)}{v_k(\xi)} - 1 \right] v_k(\xi) d\xi \right) \\
        &= 1\,.
    \end{align*}
    Under $\Q^*$, the PRM $N$ has compensator $(1+\eta) \,\nu_C(d\xi) \,dt$. Again using the representation of $Z_{k}^*$ from \zcref{prop:expression-Z-X}, we compute for $\kinI$, 
    \begin{align*}
        \E^{\Q^*}[Z^*_{k,t}] &= \exp \left(t \int_{\Rp} \!\!\!  \left[ v_k(\xi) - (1+\eta) v_C(\xi) \right] d\xi \right) \E^{\Q^*}\!\left[  \exp \left(\int_0^t \!\! \int_{\Rp} \!\!\!  \ln \left( (1+\eta) \frac{v_C(\xi)}{v_k(\xi)} \right)  N(d\xi,ds) \right)\right] \\
        &=  \exp \left(t\int_{\Rp} \!\!\!  \left[ v_k(\xi) - (1+\eta) v_C(\xi) \right] d\xi \right) \exp \left(\int_0^t \!\! \int_{\Rp} \!\!\! \left[ (1+\eta) \frac{v_C(\xi)}{v_k(\xi)} - 1 \right]  (1+\eta) \nu_C(d\xi) ds \right) \\
        &=  \exp \left( t \int_{\Rp} \!\! \left[1- (1+\eta) \frac{v_C(\xi)}{v_k(\xi)} \right]^2\!\! \nu_k(d\xi) \right) \\
        &= \ell_k(t) \,. \qedhere
    \end{align*}
\end{proof}

\zcref{prop:expression-Z-X} also allows us to find the expected wealth of the insurer under the optimal strategy and optimal measure $\Q^*$. This highlights an interesting feature of the measure $\Q^*$: under this measure, the expected value of the insurer's wealth under the optimal risk sharing strategy, $X^*$, is the same as the expected value of the insurer's wealth with no risk sharing, $X^{CL}$, given by \eqref{eqn:X_CL}.

\begin{corollary}
\label{cor:EXZ}
    Under \zcref{assump:intgr_comp}, for $t \in [0,T]$,
    \begin{align*}
        \E^{\Q^*}[X^*_t] &= x_{0} + \left[ c - (1+\eta) \int_{\Rp} \!\!\! \xi \, \nu_C(d\xi) \right] t  = \E^{\Q^*}[X^{CL}_t]\,.
    \end{align*}  
\end{corollary}

\begin{proof}
    Using \zcref{cor:EZ} and the representation of $X^*$ from \zcref{prop:expression-Z-X}, the computation of the expected value of $X$ under $\Q^*$ is straightforward as $\ell_k(t) \ell_k(-t) = \ell_k(0) = 1$:
    \begin{align*}
        \E^{\Q^*}[X^*_{t}] &= x_{0} + \left[ c - (1+\eta) \int_{\Rp} \!\!\! \xi \, \nu_C(d\xi) \right] t + \frac{1}{\theta} \sum_{\kinI} \pi_k \, \ell_k(T) \, \left[1 - \ell_k(-t) \, \E^{\Q^*}[Z^*_{k,t}]\right] \\
        &= x_{0} + \left[ c - (1+\eta) \int_{\Rp} \!\!\! \xi \, \nu_C(d\xi) \right] t + \frac{1}{\theta} \sum_{\kinI} \pi_k \, \ell_k(T) \, \left[1 - \ell_k(-t) \ell_k(t) \right] \\
        &= x_{0} + \left[ c - (1+\eta) \int_{\Rp} \!\!\! \xi \, \nu_C(d\xi) \right] t  \,. 
    \end{align*}
    To see that this is the same as the expected wealth with no risk sharing, note that 
    \begin{equation*}
        \E^{\Q^*}\left[X^{CL}_{t}\right] = \E^{\Q^*}\!\left[x_{0} + ct -  \int_0^t \!\! \int_{\Rp} \!\!\! \xi  \, N(d\xi,ds)\right] = x_{0} + c\,t -  (1+\eta) \, t \int_{\Rp} \!\!\! \xi  \,\nu_C(d\xi)  \,.\qedhere
    \end{equation*}
\end{proof}

The variance of the insurer's wealth under the optimal strategy may also be computed under $\Q^*$.

\begin{proposition}
\label{prop:Cov-Z}
Under \zcref{assump:intgr_comp_jk}, we have that for $t \in (0,T]$ and $j,k\in\cI$,
\begin{equation*}
     \mathrm{Cov}^{\Q^*}(Z^*_{k,t},Z^*_{j,t}) = \exp \left( t \int_{\Rp} \!\!\!  \left[ v_k(\xi) +  v_j(\xi)  - 3 (1+\eta) v_C(\xi) + (1+\eta)^3 \frac{v_C^3(\xi)}{v_k(\xi)v_j(\xi)}\right]  \,d\xi \right) - \ell_k(t) \ell_j(t) \,.
\end{equation*}
Let $\Sigma^{\Q^*}_{\bZ^*}$ denote the covariance matrix of $\bZ^*$, i.e., $(\Sigma^{\Q^*}_{\bZ^*})_{jk} = \mathrm{Cov}^{\Q^*}(Z^*_{j,t},Z^*_{k,t})$ and $\boldsymbol{p}_t := (\pi_1 \ell_1(T-t), \dots ,\pi_n \ell_n(T-t),\pi_C \ell_C(T-t))$. Then
\begin{align*}
    \var{\Q^*} (X^*_t) = \frac{1}{\theta^2} \, 
    \boldsymbol{p}_t^\intercal \,
    \Sigma^{\Q^*}_{\bZ^*} \, \boldsymbol{p}_t \,.
\end{align*}
\end{proposition}

\begin{proof}
The result again follows from \zcref{prop:expression-Z-X}. Using the exponential formula for Poisson random measures and \zcref{assump:intgr_comp_jk}, we have for $j,k\in\cI$,
\begin{align*}
        \E^{\Q^*} [Z^*_{j,t} Z^*_{k,t}] &= \exp \left(t \int_{\Rp} \!\!\!  \left[ v_j(\xi) + v_k(\xi) - 2(1+\eta) v_C(\xi) \right] d\xi \right) \\
        & \qquad \times \E^{\Q^*}\!\left[  \exp \left(\int_0^t \!\! \int_{\Rp} \!\!\!  \ln \left( (1+\eta)^2 \frac{v_C^2(\xi)}{v_j(\xi) v_k(\xi)} \right)  N(d\xi,ds) \right)\right] \\
        &=  \exp \left(t \int_{\Rp} \!\!\!  \left[ v_j(\xi) + v_k(\xi) - 3(1+\eta) v_C(\xi) + (1+\eta)^3 \frac{v_C^3(\xi)}{v_j(\xi) v_k(\xi)} \right] d\xi \right) \,,
    \end{align*}
which, combined with the result from \zcref{cor:EZ}, gives the covariance. 
Then using the representation of $X^*$, we have
\begin{equation*}
    \var{\Q^*}\!(X^*_t) = \var{\Q^*} \! \left( \frac{1}{\theta} \sum_{\kinI} \pi_k \, \ell_k(T-t)  Z^*_{k,t}\right) = \frac{1}{\theta^2} \!\! \sum_{j, k, \in \cI} \! \! \pi_j \pi_k \ell_j(T-t) \ell_k(T-t) \mathrm{Cov}^{\Q^*}\!(Z^*_{k,t},Z^*_{j,t}) \,,
\end{equation*}
which gives the result. 
\end{proof}

We conclude this section with a remark on the effect of the model penalization parameter, $\theta$.
\begin{remark}
    We are interested in particular in the effect of the model penalization parameter, $\theta$, on the insurer's wealth under the optimal strategy, $X^*$. The result of the previous proposition shows that, like in a traditional MMV setting, $\theta$, acts as a variance penalty. As $\theta$ increases, all else fixed, the variance of $X^*$ under $\Q^*$ is reduced. A similar statement holds for the variance of $X^*$ under $\P_C$. In fact, under any measure $\P_k$, $\kinI$, we have $\var{\P_k} (X^*_t) = \frac{1}{\theta^2} \, 
    \boldsymbol{p}_t^\intercal \,
    \Sigma^{\P_k}_{\bZ^*} \, \boldsymbol{p}_t$, as long as the covariance matrix of $\bZ^*$ under $\P_k$, $\Sigma^{\P_k}_{\bZ^*}$, is defined. 
    
    Furthermore, under $\Q^*$, $\theta$ does not affect the mean of $X^*$. However, this is not true under other probability measures. For example, it is straightforward to show (see the proof of \zcref{prop:EY}) that under $\P_C$, the mean of $X_t^*$ is 
    \begin{equation*}
        \E^{\P_C}[X^*_t] = x_{0} + t \left[ c - (1+\eta) \int_{\Rp} \!\!\! \xi \, \nu_C(d\xi) \right] + \frac{1}{\theta} \sum_{\kinI} \pi_k \, \ell_k(T) \, \left[1 - \exp \left( t \, \eta \int_{\R_+} \!\! \left( 1 - (1+\eta) \frac{v_C(\xi)}{v_k(\xi)} \right) v_C(\xi) \, d\xi  \right)\right]  \,.
    \end{equation*}
\end{remark}

\subsection{Verification}
\label{sec:verification}

Next, we show that the candidate value function is indeed the value function for \zcref{opt:insurer-problem} and confirm the optimality of the candidate controls. First, we state a verification theorem. This theorem follows from that of \textcite[Theorem 3.2]{mataramvuraoksendal2008}, therefore we omit the proof (see also \textcite{TrybulaZawisza2019} for a similar approach). We then show that the proposed candidate controls and candidate value function satisfy the verification theorem. 

\begin{theorem}[Verification]
\label{thm:verification}
    Suppose there exists a continuously differentiable function $\varphi$ on  $(0,T) \times \R \times \Rpnpo$ that is continuous on $[0,T] \times \R \times [0,\infty)^{n+1}$ and Markovian controls $(\hat\alpha, \hat\beta) \in (\A, \mfB)$ such that the following hold:
    \begin{align}
        A^{\hat \alpha, \, \beta} \varphi(t,x,\bz) \geq 0 \quad &\text{ for all } \beta \in \mfB, \, (t,x,\bz) \in [0,T] \times \R \times \Rpnpo \,, \label{eqn:ver1}
        \\
        A^{\alpha, \, \hat \beta} \varphi(t,x,\bz) \leq 0 \quad &\text{ for all } \alpha \in \A,  \, (t,x,\bz) \in [0,T] \times \R \times \Rpnpo \,,  \label{eqn:ver2}
        \\
        A^{\hat \alpha, \, \hat \beta} \varphi(t,x,\bz) = 0 \quad &\text{ for all } (t,x,\bz) \in [0,T] \times \R \times \Rpnpo \,,  \label{eqn:ver3}
        \\
        \varphi(T, x ,\bz ) = x + \frac{1}{2\theta} \sum_{\kinI} \pi_k(z_k -1) \quad &\text{ for all } (x,\bz) \in \R \times \Rpnpo  \text{, and }  \label{eqn:ver4}
        \\
        \E^{\Q_\beta}_{t,x,\bz}\left[ \sup_{s\in[t,T]} \left| \varphi(s, X^{\alpha}_s, Z^{\beta}_s) \right| \right] < \infty  \quad &\text{ for all } \alpha \in \A,\, \beta \in \mfB, \, (t,x,\bz) \in [0,T] \times \R \times \Rpnpo \,, \label{eqn:ver5}
   \end{align}
   where $A^{a,b}$ is given by \eqref{eqn:generator}.
   Then 
   \begin{equation*}
       J^{\alpha,\hat\beta} \! (t,x,\bz) \leq \varphi(t,x,\bz) \leq J^{\hat\alpha,\beta}(t,x,\bz)  \quad \text{ for all } \alpha \in \A, \beta \in \mfB \,,
   \end{equation*}
    $\hat\alpha, \, \hat\beta$ are optimal controls, and $\varphi(t,x,\bz) = J^{\hat\alpha,\hat\beta}(t,x,\bz)$.
\end{theorem}

We first show that the candidate optimal controls given in \zcref{prop:optim_contr} are indeed admissible.

\begin{proposition}
\label{prop:alpha_adm}
    The candidate optimal controls \[
    \alpha^*(t,\xi,\bZ^*_t) = \xi  - \frac{1}{\theta} \sum_{\kinI} \pi_k \, Z^*_{k,t} \, \ell_k(T-t) \left[ (1+\eta) \frac{v_C(\xi)}{v_k(\xi)} - 1 \right]
    \]
    and $\beta^*(\xi) = (1+\eta) \, v_C(\xi)$ are admissible.
\end{proposition}
\begin{proof}
    First note that for $k\in\cI$ and $t \in [0,T]$, we have
    \begin{align}
        \E^{\P_C}\left[ \left(Z^*_{k,t} \right)^2\right] &= \exp \left(2t \int_{\Rp} \!\!\!  \left[ v_k(\xi) -(1+\eta) v_C(\xi) \right] d\xi \right) \E^{\P_C}\!\!\left[  \exp \left(\int_0^t \!\! \int_{\Rp} \!\!\!  \ln \left( (1+\eta)^2 \frac{v_C^2(\xi)}{v_k^2(\xi)} \right)  N(d\xi,ds) \right)\right] \nonumber \\
        &=  \exp \left(t \int_{\Rp} \!\!\!  \left[ 2 v_k(\xi) - (3+2\eta) v_C(\xi) + (1+\eta)^2 \frac{v_C^3(\xi)}{v_k^2(\xi)} \right] d\xi \right) < \infty \label{eqn:ZsquaredPC}  \,,
    \end{align}
    where the second equality follows by the exponential formula for PRMs and the inequality by the integrability of the compensators $v_k(\xi)$, $\kinI$,
    and \zcref{assump:intgr_comp_jk}.

    To show that $\alpha^*$ is admissible, we check the two conditions from \zcref{def:alpha_int}. For the first, observe that
    \begin{align*}
        &\E^{\P_C}\left[ \int_0^T \!\!\! \int_{\R_+} \!\!\! | \alpha^*(s,\xi,\bZ^*_s) |^2 \,  \nu_C(d\xi) ds \right] \\
        &\qquad= \E^{\P_C}\left[ 
        \int_0^T \!\!\! \int_{\R_+} \!\! \left( \xi   - \frac{1}{\theta} \sum_{\kinI} \pi_k \, \ell_k(T-s) Z_{k,s}^* \left[ (1+\eta) \frac{v_C(\xi)}{v_k(\xi)} - 1 \right] \right)^2 \! \nu_C(d\xi) ds \right]
        \\
        &\qquad\leq 2 \, \E^{\P_C}\Bigg[ 
        \int_0^T \!\!\! \int_{\R_+} \!\!\! \xi^2 \, \nu_C(d\xi) ds + \int_0^T \!\!\! \int_{\R_+} \!\!\! \left( \frac{1}{\theta} \sum_{\kinI} \pi_k \, \ell_k(T-s) Z_{k,s}^* \left[ (1+\eta) \frac{v_C(\xi)}{v_k(\xi)} - 1 \right] \right)^2 \! \nu_C(d\xi) ds \Bigg]
        \\
        &\qquad\leq 2 \, T \int_{\R_+} \!\!\! \xi^2 \, \nu_C(d\xi) + \frac{2}{\theta^2} 
        \int_0^T
        \left(\sum_{\kinI} \E^{\P_C}\left[(\pi_k \, \ell_k(T-s) Z_{k,s}^*)^2
        \right]
        \right)\,ds\;
        \int_{\R_+} \!\!
        \left(\sum_{\kinI} \left((1+\eta) \frac{v_C(\xi)}{v_k(\xi)} - 1\right)^2
        \right)
        \nu_C(d\xi),
    \end{align*}
    where the last inequality follows from the Cauchy-Schwartz inequality. We have $\int_{\R_+} \! \xi^2 \, \nu_C(d\xi) < \infty$ by assumption, and further by \zcref{assump:intgr_comp}, \zcref{assump:intgr_comp_jk}, and the integrability of $v_C(\xi)$, we have that
    \[
        \int_{\R_+} \!\!
        \left(\sum_{\kinI} \left((1+\eta) \frac{v_C(\xi)}{v_k(\xi)} - 1\right)^2
        \right)
        \nu_C(d\xi) < \infty\,.
    \]
    Moreover, from \eqref{eqn:ZsquaredPC}, for each $\kinI$, $\E^{\P_C}[(Z^*_{k,t})^2]=e^{t \, a_k}$ for some finite constant $a_k$ and $\ell_k(t)=e^{t\,b_k}$ for some finite constant $b_k$ (see \zcref{prop:optim_contr}), therefore,
    \[
    \frac{2}{\theta^2} 
        \int_0^T \!
        \left(\sum_{\kinI} \pi_k^2 \, \ell_k^2(T-s) \E^{\P_C}\left[(Z_{k,s}^*)^2
        \right]
        \right)\,ds < \infty \,.
    \]
    Putting the inequalities together, we obtain
    \[
        \E^{\P_C}\left[ \int_0^T \!\!\! \int_{\R_+} \!\!\! | \alpha^*(s,\xi,\bZ^*_s) |^2 \,  \nu_C(d\xi) ds \right] < \infty \,.
    \]
    
    For the second condition of \zcref{def:alpha_int}, we have
    \begin{align*}
        &\E^{\P_C}\Bigg[ \int_0^T \!\!\! \int_{\R_+} \!\!\!\! \big[ \xi - \alpha^*(s,\xi,\bZ^*_s)\big]^2  \nu_C(d\xi) ds \Bigg] \\
        &\qquad= \E^{\P_C}\Bigg[ \int_0^T \!\!\! \int_{\R_+} \!\! \Bigg[\frac{1}{\theta} \sum_{\kinI} \pi_k  \, \ell_k(T-s) \, Z_{k,s}^* \left((1+\eta) \frac{v_C(\xi)}{v_k(\xi)} - 1 \right)\Bigg]^2 \! \nu_C(d\xi) ds \Bigg]  
        \\
        &\qquad\le \frac{1}{\theta^2}  \left( \sum_{k \in \cI} \pi_k^2 \int_0^T \!\! \ell_k^2(T-s) \E^{\P_C}\!\! \left[(Z_{k,s}^*)^2 \right] ds \right)
        \left(\int_{\R_+} 
        \sum_{\kinI} \left((1+\eta) \frac{v_C(\xi)}{v_k(\xi)} - 1\right)^2 \nu_C(d\xi)
        \right)
        \\
         &\qquad < \infty  \,.
    \end{align*}
    The first inequality follows from Cauchy-Schwartz and the second inequality we established earlier.

    Furthermore, $\beta^*$ satisfies \zcref{def:beta_int2} by \zcref{cor:EZ}, as for all $\kinI$, $\E^{\P_k}  [ ( Z_{k,T}^* )^2 ] = \E^{\Q^*} \! [ Z_{k,T}^* ] = \ell_k(T) < \infty$.
    Therefore $\alpha^*$, $\beta^*$ are admissible controls. 
\end{proof}

Next, we confirm that the candidate optimal controls and value function from \zcref{prop:optim_contr} satisfy the conditions of \zcref{thm:verification}.

\begin{proposition}
\label{prop:ver}
    The candidate value function $J(t,x,\bz)$ and the candidate controls $\alpha^*$, $\beta^*$ given in \zcref{prop:optim_contr} satisfy \zcref{thm:verification}, and are therefore the optimal value functions and controls of \zcref{opt:insurer-problem}.
\end{proposition}

\begin{proof}
    By the previous proposition, $\alpha^*$ and $\beta^*$ are admissible. Recall that the candidate value function is 
    \begin{equation*}
        J(t,x,\bz) = x + \sum_{\kinI} \frac{\pi_k}{2 \theta} \, z_k \, \ell_k(T-t) - \frac{1}{2\theta} - \left[ (1+\eta) \int_{\Rp} \!\! \xi \, \nu_C(d\xi) - c \right] (T-t) \,,
    \end{equation*}
    where
    \begin{equation*}
        \ell_k(t) = \exp \left( t \int_{\Rp} \!\! \left[1- (1+\eta) \frac{v_C(\xi)}{v_k(\xi)} \right]^2\!\! \nu_k(d\xi)\right) \,.
    \end{equation*}
    By \zcref{assump:intgr_comp}, $\ell_k(t) < \infty$ for all $t \in [0,T]$ and furthermore $\ell_k(t)$ is an integrable and continuously differentiable function. Then, as $J(t,x,\bz)$ is a linear combination of continuously differentiable functions of $x, \, z_k$, and $t$, it is continuously differentiable on  $(0,T) \times \R \times \Rpnpo$ and continuous on its closure.
    
    Conditions \eqref{eqn:ver3} and \eqref{eqn:ver4} follow from the HJBI equation in the proof of \zcref{prop:optim_contr} (see \eqref{eqn:HJBI1} and \eqref{eqn:HJBI2}). Next, we show \eqref{eqn:ver1} and \eqref{eqn:ver2} hold. For Markovian $\alpha \in \A$, $\beta \in \mfB$, and $(t,x,\bz) \in [0,T] \times \R \times \Rpnpo$ we have
    \begin{align*}
        A^{\alpha,\,\beta} J(t,x,\bz) =&\, \frac{1}{\theta} \sum_{\kinI} \pi_k \, z_k \, \ell_k(T-t) \int_{\Rp} \!\! \left( \left[ 1-\frac{\beta_t(\xi,x,\bz)}{v_k(\xi)} \right]^2-\left[1- \frac{(1+\eta) v_C(\xi)}{v_k(\xi)} \right]^2 \right) \nu_k(d\xi) 
        \\
        & \; + \int_{\Rp} \!\!\! \left[\xi - \alpha_t(\xi,x,\bz) \right]\left[ (1+\eta) v_C(\xi) - \beta_t(\xi)\right] d\xi \,.
    \end{align*}
    Then, substituting in $\beta^*=(1+\eta) v_C$ from \eqref{eqn:beta-star} gives
    \begin{equation*}
        A^{\alpha,\beta^*} J(t,x,\bz) = 0 \,,
    \end{equation*}
    for $\alpha \in \A$ (Markovian) arbitrary, so \eqref{eqn:ver2} holds. Furthermore, substituting $\alpha_t^*$ from \eqref{eqn:alpha-star}, after simplifying, we find
    \begin{equation*}
        A^{\alpha^{\!*}\!, \, \beta} J(t,x,\bz) = \frac{1}{\theta} \sum_{\kinI} \pi_k \, z_k \, \ell_k(T-t) \int_{\Rp} \!\! \frac{\left[\beta_t(\xi,x,\bz)) - (1+\eta) v_C(\xi)\right]^2}{ 2\, v_k(\xi)} \, d\xi \geq 0 \,,
    \end{equation*}
    for $\beta \in \mfB$ (Markovian) arbitrary, showing \eqref{eqn:ver1} holds.

    We next show that \eqref{eqn:ver5} holds. 
    To this end, take $\alpha\in\A$ and $\beta\in\cB$ arbitrary. It then follows that
    \begin{align*}
        & \E^{\Q_\beta}_{t,x,\bz}\left[ \sup_{s\in[t,T]} \left| J(s, X^{\alpha}_s, Z^{\beta}_s) \right| \right] 
        \\ 
        &= \E^{\Q_\beta}_{t,x,\bz}\left[ \sup_{s\in[t,T]}  \left| X_s^\alpha + \sum_{\kinI} \frac{\pi_k}{2 \theta} \, \ell_k(T-s) \, Z_{k,s}^\beta -  \frac{1}{2\theta} - \left[ (1+\eta) \int_{\Rp} \!\! \xi \, \nu_C(d\xi) - c \right] (T-s)  \right| \right] 
        \\
        &\leq \E^{\Q_\beta}_{t,x,\bz}\left[ \sup_{s\in[t,T]}  \left| X_s^\alpha \right| \right] + \sum_{\kinI} \frac{\pi_k}{2 \theta} \,  \E^{\Q_\beta}_{t,x,\bz}\left[ \sup_{s\in[t,T]} \left| \ell_k(T-s) \, Z_{k,s}^\beta \right| \right]
        \\ &\qquad 
        + \frac{1}{2\theta} + \sup_{s\in[t,T]} \left( \left| (1+\eta) \textstyle\int_{\Rp} \!\! \xi \, \nu_C(d\xi) - c \right| (T-s) \right) 
        \\
        &= \E^{\Q_\beta}_{t,x,\bz}\left[ \sup_{s\in[t,T]}  \left| X_s^\alpha \right| \right] + \sum_{\kinI} \frac{\pi_k}{2 \theta} \,\ell_k(T-t) \,\E^{\Q_\beta}_{t,x,\bz} \left[ \sup_{s\in[t,T]} \left|  \, Z_{k,s}^\beta \right| \right] 
        + \frac{1}{2\theta} + \left| (1+\eta) \textstyle\int_{\Rp} \!\! \xi \, \nu_C(d\xi) - c \right| (T-t) \,.
    \end{align*}
    Thus, to show that \eqref{eqn:ver5} holds, it suffices to show that $ \E^{\Q_\beta}_{t,x,\bz}\left[ \displaystyle\sup_{s\in[t,T]}  \left| X_s^\alpha \right| \right] < \infty$ and $\E^{\Q_\beta}_{t,x,\bz} \left[ \displaystyle\sup_{s\in[t,T]} \left|  \, Z_{k,s}^\beta \right| \right] < \infty$ for all $\kinI$. Continuing, for each $\kinI$, we have that
    \begin{align*}
       \E^{\Q_\beta}_{t,x,\bz} \left[ \sup_{s\in[t,T]} \left|  \, Z_{k,s}^\beta \right| \right]  &= \E^{\P_k}_{t,x,\bz} \left[ \frac{d\Q_\beta}{d\P_k} \sup_{s\in[t,T]} \left|  \, Z_{k,s}^\beta \right| \right] 
       \\
       &\leq \left( \E^{\P_k}_{t,x,\bz} \left[ \left( \frac{d\Q_\beta}{d\P_k} \right)^2 \right] \right)^\frac12 \left( \E^{\P_k}_{t,x,\bz} \left[ \sup_{s\in[t,T]} \left|  \, Z_{k,s}^\beta \right|^2 \right] \right)^\frac12
       \\
       &\leq \left( \E^{\P_k}_{t,x,\bz} \left[ \left( \frac{d\Q_\beta}{d\P_k}\right)^2 \right] \right)^\frac12 \left( 4 \, \E^{\P_k}_{t,x,\bz} \left[  \left(  \, Z_{k,T}^\beta \right)^2 \right] \right)^\frac12
       \\
       &= 2\,\E^{\P_k}_{t,x,\bz} \left[  \left(  \, Z_{k,T}^\beta \right)^2 \right] \\
       &< \infty \,,
    \end{align*}
    where the first inequality follows from Cauchy-Schwartz and  the second from Doob's maximal inequality. The second equality follow as, by definition of the $Z_k^\beta$ processes, $Z_{k,T}^\beta=\frac{d\Q_\beta}{d\P_k}$. The final inequality follows as $\beta\in\cB$ (see \zcref{def:beta_int2}).  Furthermore, using the triangle inequality followed by Jensen's inequality, we obtain
    \begin{align}
        &\hspace*{-1em}\E^{\Q_\beta}_{t,x,\bz}\left[ \sup_{s\in[t,T]}  \left| X_s^\alpha \right| \right]
        \nonumber
        \\
        &= \E^{\Q_\beta}_{t,x,\bz}\Bigg[ \sup_{s\in[t,T]}  \Bigg| x_{0} + \int_t^s \left[ c - (1+\eta) \int_{\Rp} \!\!\! \alpha_u(\xi) \, \nu_C (d\xi) - \int_{\Rp} \!\! \left[ \xi - \alpha_u(\xi) \right] \nu_C(d\xi) \right] du
        \nonumber
        \\
        & \hspace*{7em} -  \int_t^s \!\! \int_{\Rp} \!\!\! [\xi - \alpha_u(\xi)] \, \tilde N^{\P_C}(d\xi,du) \Bigg| \Bigg] 
        \nonumber
        \\
        &\leq x_{0} + c \, (T-t) + |1+\eta|\; \E^{\Q_\beta}_{t,x,\bz}\left[ \sup_{s\in[t,T]}  \left| \int_t^s \!\! \int_{\Rp} \!\!\! \alpha_u(\xi) \, \nu_C (d\xi) du\right| \right] 
        \nonumber
        \\
        & \qquad + \E^{\Q_\beta}_{t,x,\bz}\left[ \sup_{s\in[t,T]}  \left| \int_t^s \!\! \int_{\Rp} \!\!\! \left[ \xi - \alpha_u(\xi) \right]\, \nu_C (d\xi) du\right| \right] + \E^{\Q_\beta}_{t,x,\bz}\left[ \sup_{s\in[t,T]}  \left|  \int_t^s \!\! \int_{\Rp} \!\!\! [\xi - \alpha_u(\xi)] \, \tilde N^{\P_C}(d\xi,du) \right| \right] 
        \nonumber
        \\
        &\leq x_{0} + c \, (T-t) + |1+\eta|\,  \E^{\Q_\beta}_{t,x,\bz}\left[  \sup_{s\in[t,T]} \left(\int_t^s \!\! \int_{\Rp} \!\!\!  \left| \alpha_u(\xi) \right| \, \nu_C (d\xi) du \right) \right] 
        \nonumber
        \\
        & \qquad + \E^{\Q_\beta}_{t,x,\bz}\left[ \sup_{s\in[t,T]}  
        \left(\int_t^s \!\! \int_{\Rp} \!\!\! \left| \xi - \alpha_u(\xi) \right| \, \nu_C (d\xi) du
        \right) \right] + \E^{\Q_\beta}_{t,x,\bz}\left[ \sup_{s\in[t,T]}  \left|  \int_t^s \!\! \int_{\Rp} \!\!\! [\xi - \alpha_u(\xi)] \, \tilde N^{\P_C}(d\xi,du) \right| \right] 
        \nonumber
        \\
        \begin{split}
        &\leq x_{0} + c \, (T-t) + |1+\eta|\,  \E^{\Q_\beta}_{t,x,\bz}\left[ \int_t^T \!\!\!  \int_{\Rp} \!\!\!  \left| \alpha_u(\xi) \right| \, \nu_C (d\xi) du \right] 
        \\
        & \qquad + \E^{\Q_\beta}_{t,x,\bz}\left[\int_t^T \!\!\! \int_{\Rp} \!\!\! \left| \xi - \alpha_u(\xi) \right| \, \nu_C (d\xi) du\right] + \E^{\Q_\beta}_{t,x,\bz}\left[ \sup_{s\in[t,T]}  \left|  \int_t^s \!\! \int_{\Rp} \!\!\! [\xi - \alpha_u(\xi)] \, \tilde N^{\P_C}(d\xi,du) \right| \right] \,.
        \end{split}
        \label{eqn:sup-X-ineq}
    \end{align}
    We next bound each of the three expectations appearing in \eqref{eqn:sup-X-ineq} in turn. Define $\lambda_C:=\int_{\Rp} \!\!\! \nu_C (d\xi)$, noting that $\lambda_C<\infty$ by the integrability assumption on $\nu_C$.  First, using Cauchy-Schwartz twice, we have
    \begin{align*}
        & \hspace*{-1em}
        \E^{\Q_\beta}_{t,x,\bz}\left[ \int_t^T \!\!\!  \int_{\Rp} \!\!\!  \left| \alpha_u(\xi) \right| \, \nu_C (d\xi) du \right] 
        \\
        &= \E^{\P_C}_{t,x,\bz}\left[ \frac{d\Q_\beta}{d\P_C} \int_t^T \!\!\!  \int_{\Rp} \!\!\!  \left| \alpha_u(\xi) \right| \, \nu_C (d\xi) du \right] 
        \\
        &\leq \left(\E^{\P_C}_{t,x,\bz}\left[ \left(\frac{d\Q_\beta}{d\P_C} \right)^2 \right]\right)^\frac12
        \left(\E^{\P_C}_{t,x,\bz}\left[ \left( \int_t^T \!\!\!  \int_{\Rp} \!\!\!  \left| \alpha_u(\xi) \right| \, \nu_C (d\xi) du \right)^2 \right]\right)^\frac12 
        \\
        &\leq \left(\E^{\P_C}_{t,x,\bz}\left[ \left(\frac{d\Q_\beta}{d\P_C} \right)^2 \right]\right)^\frac12 
        \left( \E^{\P_C}_{t,x,\bz}\left[ \int_t^T \!\!\!  \int_{\Rp} \!\!\!  \left| \alpha_u(\xi) \right|^2 \, \nu_C (d\xi) du  \right]\right)^\frac12  
        \left( (T-t) \lambda_C \right)^\frac12
        \\
        & < \infty \,,
    \end{align*}
    where, the final inequality holds as $\alpha\in\A$ (see \zcref{def:alpha_int}) and $\beta\in\cB$ (see \zcref{def:beta_int2}).
     Next, by the same reasoning, the second term is bounded:
    \begin{align*}
         & \hspace*{-1em}
         \E^{\Q_\beta}_{t,x,\bz}\left[\int_t^T \!\!\! \int_{\Rp} \!\!\! \left| \xi - \alpha_u(\xi) \right| \, \nu_C (d\xi) du\right] 
        \\
        &\leq \left(\E^{\P_C}_{t,x,\bz}\left[ \left(\frac{d\Q_\beta}{d\P_C} \right)^2  \right]\right)^\frac12\left( \E^{\P_C}_{t,x,\bz}\left[ \int_t^T \!\!\! \int_{\Rp} \!\!\! \left| \xi - \alpha_u(\xi) \right|^2 \, \nu_C (d\xi) du  \right]\right)^\frac12 
         \left( (T-t) \lambda_C \right)^\frac12
        < \infty \,.
    \end{align*}
    Finally, for the third term, we have
    \begin{align*}
        &\hspace*{-1em}\E^{\Q_\beta}_{t,x,\bz}\left[ \sup_{s\in[t,T]}  \left|  \int_t^s \!\! \int_{\Rp} \!\!\! [\xi - \alpha_u(\xi)] \, \tilde N^{\P_C}(d\xi,du) \right| \right] 
        \\
         &\leq \left(\E^{\P_C}\left[ \left( \frac{  d\Q_\beta}{d\P_C} \right)^2 \right]
         \right)^\frac12
         \left(\E^{\P_C}_{t,x,\bz} \left[ \sup_{s\in[t,T]}  \left| \int_t^s \!\!\! \int_{\Rp} \!\!\!   \left[ \xi - \alpha_u(\xi) \right] \, \tilde N^{\P_C}(d\xi,du) \right|^2  \right]\right)^\frac12
         \\
         &\leq
         K\left(\E^{\P_C}\left[ \left( \frac{  d\Q_\beta}{d\P_C} \right)^2 \right]
         \right)^\frac12
         \left(\E^{\P_C}_{t,x,\bz} \left[ \int_t^T \!\!\! \int_{\Rp} \!\!\!   \left| \xi - \alpha_u(\xi) \right|^2 \, N (d\xi,du)  \right]\right)^\frac12
         \\
         &\leq
         K\left(\E^{\P_C}\left[ \left( \frac{  d\Q_\beta}{d\P_C} \right)^2 \right]
         \right)^\frac12
         \left(\E^{\P_C}_{t,x,\bz} \left[ \int_0^T \!\!\! \int_{\Rp} \!\!\!   \left| \xi - \alpha_u(\xi) \right|^2 \, \nu_C(d\xi) du \right]\right)^\frac12
         \\ 
         &< \infty,
    \end{align*}
    where $K$ is some positive constant, the second inequality follows by the Burkholder-Davis-Gundy inequality, and the third inequality by replacing the random measure with its compensator and extending the integral. The final inequality holds as $\alpha\in\A$ (see \zcref{def:alpha_int}) and $\beta\in\cB$ (see \zcref{def:beta_int2}). 

    Combining these three inequalities with \eqref{eqn:sup-X-ineq}, we find that $\E^{\Q_\beta}_{t,x,\bz}\left[ \displaystyle\sup_{s\in[t,T]}  \left| X_s^\alpha \right| \right]<\infty$ and hence \eqref{eqn:ver5} holds.
\end{proof}

This shows that the optimal risk sharing strategy and model derived in \zcref{prop:optim_contr} are indeed optimal for the insurer's risk sharing problem.

\section{Equivalent monotone mean-variance and mean-variance formulations}
\label{sec:MMV-equiv}

In this section, we show that the model combination problem solved in \zcref{sec:main-results} has an equivalent monotone mean-variance formulation under a composite probability measure $\bar\P$, which we characterize. This composite probability measure depends on all of the reference measures and is not the same as any of them individually. This gives an alternative perspective on how the model combination enters the decision problem.
Then, given extant results in the literature showing the equivalence between monotone mean-variance and mean-variance problems (e.g., \cite{Cerny2020,Li2025ORL}), we formulate the mean-variance risk sharing problem under $\bar \P$, and show that its solution coincides with that of the model combination problem.

First we show a general equivalence between the values of the MMV problem, \eqref{eqn:def_MMV}, with respect to a probability measure $\bar\P$, and the MC problem, \eqref{eqn:def_MC}, with respect to a finite set of equivalent probability measures $\{\P_k\}_{k\in\N}$. This result applies to the model combination problem for an arbitrary finite set of equivalent probability measures and is not specific to the risk sharing problem.

\begin{theorem}
\label{thm:P_bar}
    Given a finite set of equivalent probability measures $\{\P_k\}_{k\in\N}$, fix a probability measure $\P$ which is equivalent to all $\P_k$, $k\in\N$, and define the composite probability measure $\bar \P$ by the Radon-Nikodym derivative
    \begin{equation}
    \label{eqn:Pbar_def}
        \frac{d\bar\P}{d\P} = \frac1\zeta \left( \sum_{k\in\N} \pi_k \frac{d\P\,}{d\P_k}\right)^{-1} \!, \quad \text{where } \zeta := \E^{\P} \left[ \left( \sum_{k\in\N} \pi_k \frac{d\P\,}{d\P_k}\right)^{-1}\right]
    \end{equation}
    Then, letting $\bar \theta := \zeta \theta$, we have
    \begin{equation*}
       V^{\theta}_{\text{MC}}(X;\{\P_k\}_{k \in \N}) = V^{\bar \theta}_{\text{MMV}}(X;\bar\P) + \frac{1}{2} \left( \frac{1}{\bar \theta} - \frac{1}{\theta} \right)\,.
    \end{equation*}
\end{theorem}

\begin{proof}
Let $\P$ be an arbitrary probability measure which is equivalent to all $\P_k$, $k\in\N$. Such a $\P$ exists since the $\P_k$, $k\in\N$, are equivalent.
We begin by changing the second term in the model combination criterion to an expectation under $\P$:\footnote{We gratefully acknowledge the suggestion of an anonymous reviewer which motivated this representation.}
\begin{align*}
    V^{\theta}_{\text{MC}}(X;\{\P_k\}_{k \in \N}) &:=
    \min_{\Q \in \Delta^2(\{\P_k\}_{k \in \N})} \left\{ \E^\Q  \left[ X \right]  + \frac{1}{2 \theta} \sum_{k\in\N} \pi_k \, \E^{\P_k} \left[ \left(\frac{d\Q}{d\P_k}\right)^2 - 1  \right]  \right\} \,, \\
    &=
    \min_{\Q \in \Delta^2(\{\P_k\}_{k \in \N})} \left\{ \E^\Q  \left[ X \right]  + \frac{1}{2 \theta} \sum_{k\in\N} \pi_k \, \E^{\P_k} \left[ \left(\frac{d\Q}{d\P}\right)^2  \left(\frac{d\P}{d\P_k}\right)^2 \right]  \right\}  - \frac{1}{2\theta}\,, \\
    &=
    \min_{\Q \in \Delta^2(\{\P_k\}_{k \in \N})} \left\{ \E^\Q  \left[ X \right]  + \frac{1}{2 \theta} \, \E^{\P} \left[ \left(\frac{d\Q}{d\P}\right)^2  \sum_{k\in\N} \pi_k  \frac{d\P}{d\P_k}\right] \right\}  - \frac{1}{2\theta}\,.
\end{align*}

Define a random variable $S$ which is the weighted sum of the RN densities
\begin{equation*}
     S = \sum_{k\in\N} \pi_k \frac{d\P}{d\P_k} \,.
\end{equation*}
As the measures $\P_k$ are equivalent to $\P$, and the set $\N$ is finite, we have that $0<S<\infty$ holds $\P$-almost surely. 
Furthermore, let $\zeta = \E^{\P}[1/S]$. Then $\zeta \in (0,\infty)$. To see the upper bound, let $m\in\N$ be the index of a probability measure with a strictly positive weight $\pi_m >0$. One must exist by our assumption that the weights sum to 1. Then
$S = \sum_{k\in\N} \pi_k \frac{d\P}{d\P_k} \geq \pi_m \frac{d\P}{d\P_m}$ $\P$-a.s..
Therefore $\zeta = \E^{\P} \left[\frac{1}{S} \right] \leq \frac{1}{\pi_m} \E^{\P}\left[\frac{d\P_m}{d\P}\right] =\frac{1}{\pi_m} < \infty$.

Now we define the probability measure $\bar\P$ via the RN density
\begin{equation*}
    \frac{d \bar \P}{d \P} = \frac{1}{S} \,\frac{1}{\E^{\P}[1/S]} = \frac{1}{\zeta S}\,,
\end{equation*}
which is well-defined as $0 < 1/S < \infty$ $\P$-a.s.~and $\zeta \in (0,\infty)$. 
Then the criterion is
\begin{align*}
    V^{\theta}_{\text{MC}}(X;\{\P_k\}_{k \in \N}) = &\min_{\Q \in \Delta^2(\{\P_k\}_{k \in \N})} \left\{ \E^\Q \left[ X \right]+ \frac{1}{2 \theta} \, \E^{\P} \left[ S \left(\frac{d\Q}{d\P}\right)^2  \right]  \right\}  - \frac{1}{2\theta}\,, \\
    = & \min_{\Q \in \Delta^2(\{\P_k\}_{k \in \N})} \left\{ \E^{\Q} \left[ X \right]+ \frac{1}{2 \theta} \, \E^{\bar \P} \left[ S \left(\frac{d\Q}{d\bar \P}\right)^2 \frac{d\bar \P}{d\P}  \right]  \right\}  - \frac{1}{2\theta}\,, \\
    = & \min_{\Q \in \Delta^2(\{\P_k\}_{k \in \N})} \left\{ \E^{\Q} \left[ X \right]+ \frac{1}{2 \theta \zeta} \, \E^{\bar \P} \left[\left(\frac{d\Q}{d\bar \P}\right)^2 \right] \right\}  - \frac{1}{2\theta} \\
    = & \min_{\Q \in \Delta^2(\{\P_k\}_{k \in \N})} \left\{ \E^{\Q} \left[ X \right]+ \frac{1}{2 \bar \theta} \, \E^{\bar \P} \left[\left(\frac{d\Q}{d\bar \P}\right)^2 -1 \right] \right\}  - \frac{1}{2\theta} + \frac{1}{2\bar\theta}\,.
\end{align*}

Finally, we show $\Delta^2(\{\P_k\}_{k \in \N})= \Delta^2(\bar \P):= \left\{\Q: \Q \ll \bar \P \text{ and } \E^{\bar \P} \left[\left(\tfrac{d\Q}{d\bar \P}\right)^2\right] <\infty   \right\}$, in which case the first term of the above equation is exactly $V^{\bar \theta}_{\text{MMV}}(X;\bar\P)$. 
We have that
\begin{equation}
\label{eqn:equiv_sets}
    \E^{\bar \P} \left[ \left(\frac{d\Q}{d\bar \P}\right)^{\!\!2}  \right] = \E^{\Q} \left[ \frac{d\Q}{d\P} \frac{d\P}{d\bar \P} \right]= \zeta \, \E^{\Q} \left[ \frac{d\Q}{d\P} S \right]  = \zeta \, \E^{\Q} \left[ \frac{d\Q}{d\P} \sum_{k\in\N} \pi_k \frac{d\P}{d\P_k} \right] = \zeta \sum_{k\in\N} \pi_k  \E^{\P_k} \left[   \left( \frac{d\Q}{d\P_k}\right)^{\!\!2} \right] \,.
\end{equation}
Take $\Q \in \Delta^2(\{\P_k\}_{k \in \N})$. Then $\Q \ll \P_k$ for all $k\in\N$ and hence $\Q \ll \bar \P$ since $\bar \P \sim \P_k$ for all $k\in\N$ by construction. Furthermore, 
$\E^{\P_k} \left[\left(\tfrac{d\Q}{d\P_k}\right)^2\right] <\infty $ for all $k\in\N$ and hence $\E^{\bar \P} \left[ \left(\frac{d\Q}{d\bar \P}\right)^{\!\!2}  \right]  < \infty$ by \eqref{eqn:equiv_sets}. Thus $\Q \in \Delta^2(\bar \P)$.

Take $\Q \in \Delta^2(\bar \P)$. Then $\Q \ll \P_k$ for all $k\in\N$ since $\Q \ll \bar \P$ and $\bar \P \sim \P_k$ for all $k\in\N$. Furthermore, $\E^{\bar \P} \left[ \left(\frac{d\Q}{d\bar \P}\right)^{\!\!2}  \right] < \infty$ and thus by \eqref{eqn:equiv_sets},
\begin{equation*}
     \sum_{k\in\N} \pi_k  \E^{\P_k} \left[   \left( \frac{d\Q}{d\P_k}\right)^{\!\!2} \right] = \frac{1}{\zeta} \, \E^{\bar \P} \left[ \left(\frac{d\Q}{d\bar \P}\right)^{\!\!2}  \right] < \infty \,,
\end{equation*}
implying $\E^{\P_k} \left[   \left( \frac{d\Q}{d\P_k}\right)^{\!\!2} \right] < \infty$ for each $k\in\N$. Therefore $\Q \in \Delta^2(\{\P_k\}_{k \in \N})$. The equivalence of the sets completes the proof.
\end{proof}

Furthermore, the choice of $\P$ in \zcref{thm:P_bar} does not change the composite probability measure $\bar \P$.
\begin{proposition}
\label{prop:P_bar_unique}
    The composite probability measure $\bar \P$ is independent of the choice of the probability measure $\P$.
\end{proposition}

\begin{proof}
Let $\P'$ be a probability measure which is equivalent to all $\P_k$, $k\in\N$, and distinct from $\P$. Define the composite probability measure with respect to $\P'$ as
\begin{equation*}
        \frac{d\bar\P'}{d\P'} := \frac{1}{S'} \frac{1}{\E^{\P'}[1/S']} \,, \quad \text{where } S'= \sum_{k\in\N} \pi_k \frac{d\P'}{d\P_k} \,.
\end{equation*}  

We have that $\P' \sim \P$ and 
\begin{equation*}
    \E^{\P'}[1/S'] = \E^{\P'}\left[\frac{\frac{d\P}{d\P'}}{\frac{d\P}{d \P'}\sum_{k\in\N} \pi_k \frac{d\P'}{d\P_k}}\right] = \E^{\P}\left[\frac{1}{\sum_{k\in\N} \pi_k \frac{d\P}{d\P_k}}\right] = \zeta\,.
\end{equation*}
Furthermore, the following holds $\P$-almost surely:
\begin{equation*}
    \frac{d\bar\P'}{d\P} =  \frac{d\bar\P'}{d\P'} \frac{d\P'}{d\P} = \frac{1}{S'} \frac{1}{\E^{\P'}[1/ S']} \frac{d\P'}{d\P} =  \frac{1}{\zeta} \frac{d\P'}{d\P} \frac{1}{\sum_{k\in\N} \pi_k \frac{d\P'}{d\P_k}} = \frac{1}{\zeta}  \frac{1}{\sum_{k\in\N} \pi_k \frac{d\P}{d\P_k} } = \frac{d\bar\P}{d\P} 
     \,,
\end{equation*}
which shows the probability measures coincide.
\end{proof}

We have shown that the model combination problem attains the same value as the MMV problem plus a constant, and therefore is in that sense equivalent to an MMV problem. This raises the question of whether one could directly solve the MMV problem with respect to $\bar \P$, without incorporating the $n+1$ auxiliary processes. The answer, however, is no, because the Poisson random measure $N$ is no longer compound Poisson under $\bar \P$. In fact, as the next proposition shows, the compensator of $N$ under $\bar \P$ is non-Markovian in general. By introducing additional state variables, it may be lifted to a Markov process, where those state variables are themselves conditional measure change processes. Moreover, the conditional measure change corresponding to $\bar\P$, $\bar{Z}_t:=\E^\P[\frac{d\bar\P}{d\P}|\F_t]$, does not admit a Markov structure in $\bar{Z}$ alone, unlike the conditional measure change corresponding to the reference measures $\P_k,\,\kinI$ we study earlier. Hence, we cannot treat this case as in our earlier analysis with a single reference measure chosen to be $\bar\P$. 
\begin{proposition}
\label{prop:P-bar-compensator}
    The $\bar \P$-compensator $\bar{v}$ of $N$  admits the representation
    \begin{equation*}
    \bar{v}_t(\xi) = \frac{g\left(t,v_1(\xi)\,Z^1_{t^-},\dots,v_n(\xi)\,Z^n_{t^-},v_C(\xi)\,Z^C_{t^-}\right)}{g(t,\,Z^1_{t^-},\dots,\,Z^n_{t^-},Z^C_{t^-})}  
    \end{equation*}
    where 
    \begin{equation*}
    Z_t^k:=\E^{\P}\left[\left.\frac{d\P_k}{d\P}\right|\F_t\right], \qquad \forall \kinI
    \end{equation*}
    and $g(t,z_1,\dots,z_n,z_C)$ denotes the unique solution to the PIDE
    \begin{equation*}
    \left\{
    \begin{split}
    \partial_t g(t,z_1,\dots,z_n,z_C) + \int_{\Rp} \!\! g(t,v_1(\xi)\,z_1,\dots,v_n(\xi)\,z_n,v_C(\xi)\,z_C)\,d\xi &= 0,
    \\
    g(T,z_1,\dots,z_n,z_C) &= \zeta \left(\sum_{\kinI} \pi_k\,(z_k)^{-1}\right)^{-1}.
    \end{split}
    \right.
    \end{equation*}
\end{proposition}
The proof is given in \zcref{appendixB}.

\subsection{Related mean-variance problem}
We have shown that the model combination problem is equivalent to the MMV problem with respect to $\bar \P$. It is therefore also of interest to also investigate the mean-variance problem under $\bar \P$. We expect the problems to coincide.
The related mean-variance problem is:
\begin{optimization}
\label{optim:MV}
    The insurer seeks the solution to the following problem:
    \begin{align*}
        \sup_{\alpha \in \A} \left\{ \E^{\bar \P}[X^\alpha_T] - \frac{\bar \theta}{2} \var{\bar \P} (X^\alpha_T) + \frac{1}{2} \left( \frac{1}{\bar \theta} - \frac{1}{\theta} \right) \right\} \,,
    \end{align*}
    where
    \begin{align}
    \label{eqn:SDEX_MV}
            d X_t^\alpha &= \left[ c - (1+\eta) \int_{\Rp} \!\!\! \alpha_t(\xi) \, \nu_C (d\xi)   \right] dt - \int_{\Rp} \!\!\! [\xi - \alpha_t(\xi)] \, N(d\xi,dt) \,,
        \end{align}
    with $X^\alpha_0 = x_{0}$.
\end{optimization}
Let $X^*_T$ be the optimal wealth for \zcref{optim:MV}. Then if $1-\bar \theta(X^*_T - \E^{\bar \P}[X^*_T]) $ is in $L^2 (\bar \P)$ and almost surely non-negative, the mean-variance and monotone mean-variance problems are equivalent by \textcite{Li2025ORL}. We aim to show that this condition holds, and that therefore the mean-variance formulation of the problem is also equivalent to the model combination problem.

We construct the solution to this mean-variance problem following the pre-commitment approach of \textcite{ZhouLi2000}.
We begin by solving the quadratic utility problem $\inf_{\alpha \in \A} \E^{\bar\P} \left[( X^\alpha_T - \gamma)^2 \right]$ for $\gamma\in\R$. The key insight is to solve the problem under the probability measure $\Q^*$ instead of $\bar \P$, and expand the state space to incorporate the same auxiliary processes as in the MMV problem. By \zcref{thm:P_bar} and \zcref{prop:P_bar_unique}, we have
\begin{align*}
        \inf_{\alpha \in \A} \E^{\bar\P} \left[( X^\alpha_T - \gamma)^2 \right]  &=  \inf_{\alpha \in \A} \E^{\Q^*} \left[\frac{d\bar\P}{d\Q^*}( X^\alpha_T - \gamma)^2 \right] =  \inf_{\alpha \in \A} \frac{1}{\zeta} \E^{\Q^*} \left[\frac{( X^\alpha_T - \gamma)^2}{\sum_{\kinI} \pi_k \frac{d\Q^*}{d\P_k}} \right]  =  \inf_{\alpha \in \A} \frac{1}{\zeta} \E^{\Q^*}\left[\frac{( X^\alpha_T - \gamma)^2 }{\sum_{\kinI} \pi_k Z^*_{k,T}}\right]\,.
\end{align*}

As the following result shows, the optimal mean-variance strategy depends not just on the auxiliary processes but also on the increment $\hat X_t^*-x_0$, and therefore explicitly on the initial value $x_0$. However, the strategy may be rewritten in terms of the auxiliary processes only, eliminating the explicit wealth dependence. 

\begin{theorem}
\label{prop:MV_prob}
The optimal strategy for \zcref{optim:MV} is
    \begin{equation*}
        \hat\alpha^*(t,\hat X^*_t,\bZ^*_t) = \xi + \left( \hat X^*_t - A(t) \right) \left( \frac{\sum_{\kinI} \pi_k \ell_k(T-t) Z^*_{k,t}\left[ \frac{(1+\eta) v_C(\xi)}{v_k(\xi)}-1\right]}{\sum_{\kinI} \pi_k \ell_k(T-t) Z^*_{k,t}} \right)
    \end{equation*}
    where $\hat X^*_t$ solves \eqref{eqn:SDEX_MV} under the strategy $\hat \alpha$, $\bZ^*_t$ is given in \zcref{prop:expression-Z-X}, and 
    \begin{equation*}
        A(t) = \frac{\sum_{\kinI} \pi_k \ell_k(T)}{\theta} + x_0 - \left[ (1+\eta) \int_{\Rp}\!\!\! \xi\, v_C(\xi)d\xi - c\right]t \,.
    \end{equation*}
    Moreover, $\hat \alpha^*$ reduces to
    \begin{align*}
        \hat\alpha^*(t,\hat X_t^*,\bZ_t^*)= \xi - \frac{1}{\theta} {\sum_{\kinI} \pi_k \ell_k(T-t) Z^*_{k,t}\left[ \frac{(1+\eta) v_C(\xi)}{v_k(\xi)}-1\right]} \,.
    \end{align*}
\end{theorem}
The proof is given in \zcref{appendixA}. This shows that the optimal strategy for the mean-variance problem coincides with that of the model combination problem given in \zcref{prop:optim_contr}. Moreover, given \zcref{thm:P_bar}, both are equivalent to the monotone mean-variance problem with respect to $\bar\P$.  We confirm the equivalence of the mean-variance and monotone mean-variance problems by checking the condition of \textcite{Li2025ORL}.
\begin{proposition}
    The mean-variance problem and the monotone mean-variance problem with respect to $\bar\P$ are equivalent. 
\end{proposition}
\begin{proof}
    By \zcref{prop:XA} and the definition of $A(t)$,
\begin{equation*}
    \frac{X_T^* - A(T)}{\sum_{\kinI} \pi_k Z^*_{k,T}} = \frac{x_0 - A(0)}{\sum_{\kinI} \pi_k \ell_k(T)} = -\frac{1}{\theta} \,,
\end{equation*}
so $X_T^* =-\frac{1}{\theta}\sum_{\kinI} \pi_k Z^*_{k,T} + A(T)$.
By the definition of $\bar \P$ in \eqref{eqn:Pbar_def} and \zcref{prop:P_bar_unique}, we have
\begin{equation*}
    \frac{d\bar\P}{d\Q^*} = \frac{1}{\zeta}\frac{1}{\sum_{\kinI} \pi_k  \frac{d\Q^*}{d\P_k}} =  \frac{1}{\zeta} \frac{1}{\sum_{\kinI} \pi_k  Z^*_{k,T}}\,,
\end{equation*}
so, rearranging, we have
\begin{equation*}
    \E^{\bar \P}\left[\sum_{\kinI} \pi_k Z^*_{k,T} \right] = \frac1\zeta \,.
\end{equation*}
Therefore
\begin{equation*}
     X_T^* - \E^{\bar \P}[ X_T^* ] = -\frac{1}{\theta}\sum_{\kinI} \pi_k Z^*_{k,T} + A(T) + \frac{1}{\zeta} \frac{1}{\theta} - A(T) = -\frac{1}{\zeta \theta} \left(\zeta \sum_{\kinI} \pi_k Z^*_{k,T} -1 \right) \,.
\end{equation*}

Thus the condition from \textcite{Li2025ORL} is satisfied, as expected:
\begin{align*}
    1-\bar \theta(X^*_T - \E^{\bar \P}[X^*]) = \zeta \sum_{\kinI} \pi_k Z^*_{k,T} = \frac{d\Q^*}{d\bar\P}\,,
\end{align*}
which is almost surely non-negative and in $L^2 (\bar \P)$, showing equivalence of the MMV and MV problems under $\bar \P$.
\end{proof}

\section{The case with one reference model}
\label{sec:one-model}

In this section, we restrict our attention to a special case of the risk sharing problem in \zcref{opt:insurer-problem}. We consider the case where the insurer only has a single reference model. This is of interest as then the insurer's criterion reduces to the MMV criterion of \textcite{Maccheroni2006,Maccheroni2009} under a single reference measure, $\P$, instead of the composite measure, $\bar \P$.

Let $\P$ denote the single model available to the insurer, a probability measure on the completed and filtered measurable space $(\Omega, \F, \mathbb{F}=(\F_t)_{t \in [0,T]})$. As there is only one model, both the insurer and the counterparty use this model. As above, we assume the PRM $N$ has $\P$-compensator $\nu(d\xi,dt) = \nu(d\xi) dt = v(\xi) d\xi dt$ and define the $\P$-compensated PRM by 
\begin{equation*}
    \tilde N^{\P}(d\xi,dt) = N(d\xi,dt) - \nu(d\xi)dt \,.
\end{equation*}

The insurer's wealth process' dynamics are unchanged, and given by \eqref{eqn:X_SDE}. We  introduce an auxiliary process $Z$ as the solution to the SDE:
\begin{equation}
\label{eqn:Z_1M}
    dZ_{t}^\beta = - Z_{t^-}^\beta \int_{\Rp} \!\! \left[ 1 - \frac{\beta_t(\xi)}{v(\xi)}\right] \tilde N^{\P}(d\xi,dt)\,, \quad Z^\beta_{0} = 1 \,,
\end{equation}
where $\beta_t$ is an $\bF$-predictable random field. The risk sharing strategy $\alpha$ and the compensator $\beta$ must satisfy the following simplified version of the previous definitions, namely \zcref{def:alpha_int} with $\P_C = \P$ and $\nu_C = \nu$ and \zcref{def:beta_int2} with $\P_k=\P$ and $Z_k^\beta=Z^\beta$.

Given this set-up, we solve a simplified version of \zcref{opt:insurer-problem} where there is only one reference measure. To do so, 
we restrict the result in \zcref{prop:optim_contr} to the case $n=1$ and $v_C(\xi) = v(\xi)$. To simplify notation, in this section we assume that $\nu(d\xi)$ is compound Poisson with an arrival rate $\lambda$ and mean severity $\mu$, i.e.,
\begin{equation}
\label{eqn:lambda_asmpt}
  \lambda:=\int_{\Rp} \!\!\! \nu (d\xi) \quad \text{and} \quad  \mu := \frac{1}{\lambda}  \int_{\Rp} \!\!\! \xi \, \nu(d\xi)  \,,
\end{equation}
noting that $\lambda,\mu<\infty$ by the assumption in \eqref{eqn:nu_int_assumption}. We then consider the following problem with the monotone mean-variance criterion:

\begin{optimization}
\label{opt:insurer-problem-1M}
    The insurer seeks the solution to the following problem:
    \begin{equation*}
        \sup_{\alpha \in \A} \newinf_{\beta \in \mfB} \E^{\Q_\beta}  \left[ X^\alpha_T + \frac{1}{2 \theta} \left(Z^\beta_{T} - 1 \right)\right] \,,
    \end{equation*}
    where $X^\alpha$ has dynamics given in \eqref{eqn:X_Q_SDE} with $\nu_C=\nu$ and $X^\alpha_0 = x_0$ and $Z^\beta$ has dynamics given in \eqref{eqn:Z_1M}.
\end{optimization}

\begin{corollary}
\label{corr:optim_alpha_1model}
    The optimal controls for \zcref{opt:insurer-problem-1M} in feedback form are 
    \begin{align*}
        \alpha^*(t,\xi,z) &= \xi - \frac{ \eta z}{\theta} e^{\lambda \eta^2 (T-t)}\,, \\
        \beta^*(\xi) &= (1+\eta)v(\xi) \,,
    \end{align*}
    and the insurer's value function is
    \begin{equation*}
        \Phi(t,x,z) = x + \frac{1}{2\theta} \left[ e^{\lambda \eta^2 (T-t)} z - 1\right] - (T-t) \left[ (1 + \eta) \lambda \mu - c \right] \,.
    \end{equation*}
\end{corollary}
\begin{proof}
    This follows immediately from \zcref{prop:optim_contr} when $n=1$ and $v_C(\xi) = v_1(\xi)=v(\xi)$.
\end{proof}

The optimal compensator $\beta^*$ is now a distortion of the single ($\P$-)compensator, $v$. Similar to before, the change to the compensator under the optimal measure $\Q^*$ is to multiply the original by $1+\eta$. The optimal risk sharing strategy $\alpha^*$ is still of the form $\xi - d$, but now this $d$ is loss independent and always positive. This means that in this setting we must have $\alpha^*(t,\xi,z) < \xi$. However, $\alpha^*$ is still not a traditional reinsurance arrangement, as we may have $\alpha^*(t,\xi,z) < 0$.

By \zcref{prop:expression-Z-X}, we obtain expressions for the optimal processes $X^*$, $Z^*$ in terms of a Poisson process.
\begin{corollary}
    \label{prop:XZ1d}
    Let $M_t = \int_0^t \! \int_{\Rp} \!\! N(d\xi,ds)$. Then for $t \in [0,T]$,
    \begin{align*}
        Z_t^* &= e^{- \eta \lambda t} (1+\eta)^{M_t}  \,, \\
        X_t^* &= x_{0} - \left[ (1 + \eta) \lambda \mu - c \right] t + \frac{1}{\theta} e^{\lambda \eta^2 T}  \left( 1 - e^{-\lambda \eta^2 t}  Z_t^* \right) \,.
    \end{align*}
\end{corollary}

Between loss events, the paths of $Z^*$ are exponentially decreasing over time according to the function $e^{- \eta \lambda t}$. When a jump arrives in the process $M_t$, the path is multiplied by $1+\eta$, leading to an upward jump. $X^*$ remains linear in the auxiliary process, depending negatively on $Z^*$. The jumps in $X^*$ are driven entirely by the jumps in $Z^*$, and are always downwards.

We also compute the expectation, variance, and covariance of the processes $X^*$ and $Z^*$ under both the optimal measure $\Q^*$ and the (single) reference model $\P$. The results follow directly from \zcref{prop:XZ1d}, using the moment-generating function of the Poisson distribution. 

\begin{corollary} 
    For $t \in (0,T]$,
    \begin{alignat*}{2}
        & \E^{\P}[Z_t^*] = 1,  && \E^{\Q^*}[Z_t^*] = e^{\lambda \eta^2 t} \,,\\
        & \E^{\P}[X_t^*] =  x_{0} - \left[ (1 + \eta) \lambda \mu - c \right] t + \frac{1}{\theta} e^{\lambda \eta^2 T} \left( 1 -  e^{-\lambda \eta^2 t} \right), \hspace{1.5em} && \E^{\Q^*}[X_t^*] =  x_{0} - \left[ (1 + \eta) \lambda \mu - c\right] t, \\
        & \mathrm{Var}^{\P}(Z_t^*) = e^{\lambda \eta^2 t} - 1, && \mathrm{Var}^{\Q^*}(Z_t^*) = e^{2 \lambda \eta^2 t} \left( e^{\lambda \eta^2 (1+\eta)t} - 1 \right),\\
        & \mathrm{Var}^{\P}(X_t^*) = \frac{1}{\theta^2} e^{2\lambda \eta^2 (T-t)} \left( e^{\lambda \eta^2 t} - 1 \right), && \mathrm{Var}^{\Q^*}(X_t^*) = \frac{1}{\theta^2} e^{2 \lambda \eta^2 T} \left( e^{\lambda \eta^2  (1+\eta) t} - 1 \right), \\
        & \mathrm{Cov}^{\P}(X_t^*,Z_t^*) = -\frac{1}{\theta} e^{\lambda \eta^2 (T-t)} \left[ e^{\lambda \eta^2 t} -  1 \right], && \mathrm{Cov}^{\Q^*}(X_t^*,Z_t^*) = -\frac{1}{\theta} e^{\lambda \eta^2 (T+t)} \left( e^{\lambda \eta^2  (1+\eta) t} - 1 \right),\\
        & \mathrm{Corr}^{\P}(X_t^*,Z_t^*) = -1, && \mathrm{Corr}^{\Q^*}(X_t^*,Z_t^*) = -1.
    \end{alignat*}
\end{corollary}

Several observations are apparent given these expressions.
First, as the insurer's wealth process $X^*$ is linear in the (single) auxiliary process $Z^*$, they are perfectly negatively correlated under both probability measures. Second, as in the model combination setting, the effect of the parameter $\theta$ on the variance of $X^*$ is clear: a larger penalty $\theta$ reduces the variance under both probability measures through the $1/\theta^2$ term. Finally, we also observe that the variance of both $X^*$ and $Z^*$ is larger under $\Q^*$ than under $\P$.

\section{The counterparty's perspective}
\label{sec:counterparty}

In the preceding sections, we have derived the optimal risk sharing contract from the insurer's perspective. We have yet, however, to address whether it is in the interests of the counterparty to offer such a contract. In this section, returning to the general setup, we address this question and turn our attention to the counterparty in the risk sharing agreement. 

Recall that the counterparty has their own model for the losses, given by $\P_C$. We consider how the counterparty would optimally choose the safety loading $\eta$, i.e., the price to charge, under their model. Our goal is to show that the counterparty can price this bespoke risk sharing contract in the risk-neutral setting. 

As described in \zcref{sec:21}, the counterparty accepts the ceded loss $\alpha_t(\cdot)$ and charges the risk sharing premium
$p_C = (1+\eta) \int_{\R_+} \!\alpha_t(\xi) \,\nu_C (d\xi)$. Suppose the counterparty has a given initial wealth $y_{0} > 0$. Then for a given ceded loss process $(\alpha_t(\cdot))_{t\in[0,T]}$, the counterparty's wealth process evolves as follows:
\begin{equation*}
    Y_t^{\eta,\alpha} = y_{0} + (1+\eta) \int_0^t \!\! \int_{\Rp} \!\!\! \alpha_s(\xi) \, \nu_C (d\xi) \, ds -  \int_0^t \!\! \int_{\Rp} \!\!\! \alpha_s(\xi) \, N(d\xi,ds) \,.
\end{equation*}

\subsection{Model combination setting}

Assuming the insurer responds optimally to a fixed price $\eta$, the counterparty's wealth is 
\begin{equation}
\label{eqn:Y}
    Y^{\eta}_t = y_{0} +  (1+\eta)  \int_0^t \!\! \int_{\Rp} \!\!\! \alpha^*(s,\xi,\bZ^*_s ) \, \nu_C(d\xi) \, ds -  \int_0^t \!\! \int_{\Rp} \!\!\! \alpha^*(s,\xi,\bZ^*_s ) \, N(d\xi,ds) \,,
\end{equation}
where 
\begin{align*}
        \alpha^*(t,\xi,\bz) &= \xi  - \frac{1}{\theta} \sum_{\kinI} \pi_k \, z_k \, \ell_k(T-t) \left[ (1+\eta) \frac{v_C(\xi)}{v_k(\xi)} - 1 \right] \,.
\end{align*}
The counterparty then chooses the constant safety loading $\eta\in\R_+$ that maximizes their expected wealth under their model, $\P_C$:
\begin{optimization}
\label{opt:counterpartysproblem1}
    The counterparty seeks the solution to the following problem:
    \begin{equation*}
         \sup_{\eta \in \R_+} \E^{\P_C}[Y_T^\eta] \,,
    \end{equation*}
    where $Y^{\eta}_t$ evolves according to \eqref{eqn:Y}.
\end{optimization}
The next proposition gives an expression for this expectation, leveraging results from \zcref{subsec:processes} that provide expressions for the processes $X^*$ and $\bZ^*$.
\begin{proposition}
    \label{prop:EY}
    Under \zcref{assump:intgr_comp}, for $t \in [0,T]$ and $\eta \in \Rp$:
    \begin{equation*}
         \E^{\P_C}[Y_t^\eta] = y_{0} + t \, \eta \int_{\Rp} \!\!\! \xi \, \nu_C(d\xi) -  \frac{1}{\theta} \sum_{\kinI} \pi_k \, \ell_k(T) \, \left[1 - \exp \left( t \, \eta \int_{\R_+} \!\! \left( 1 - (1+\eta) \frac{v_C(\xi)}{v_k(\xi)} \right) v_C(\xi) \, d\xi  \right)\right] \,.
    \end{equation*}
\end{proposition}
\begin{proof}
    Under $\P_C$, the PRM $N$ has compensator $\nu_C(d\xi) \,dt$. By a similar argument to \zcref{cor:EXZ}, we can show that for $\kinI$, 
    \begin{align*}
        \E^{\P_C}[Z^*_{k,t}] &= \exp \left( t \int_{\Rp} \!\! \left[ \left(1- \frac{v_C(\xi)}{v_k(\xi)} \right)^2\!\! v_k(\xi) + \eta \left(\frac{v_C(\xi)}{v_k(\xi)} - 1 \right) v_C(\xi) \right] d\xi\right) \,,
    \end{align*}
    and thus
    \begin{equation*}
        \E^{\P_C}[X^*_t] = x_{0} + t \left[ c - (1+\eta) \int_{\Rp} \!\!\! \xi \, \nu_C(d\xi) \right] + \frac{1}{\theta} \sum_{\kinI} \pi_k \, \ell_k(T) \, \left[1 - \exp \left( t \, \eta \int_{\R_+} \!\! \left( 1 - (1+\eta) \frac{v_C(\xi)}{v_k(\xi)} \right) v_C(\xi) \, d\xi  \right)\right]  \,. 
    \end{equation*}
    We observe that $Y_t$ is the difference between the insurer's uncontrolled wealth process $X_t^{\text{CL}}$, given by \eqref{eqn:X_CL}, and the insurer's wealth process under the optimal control, $X_t^*$, given by \eqref{eqn:X_SDE} evaluated at $\alpha^*$, plus the counterparty's initial wealth, $y_0$. As under $\P_C$, the expected value of $X_t^{\text{CL}}$ is $x_{0}+t\left(c -  \int_{\Rp} \!\!\! \xi \, \nu_C(d\xi) \right)$, we have
    \begin{align*}
        \E^{\P_C}[Y_t^\eta] &= y_{0} + \E^{\P_C}[X_t^{\text{CL}}] - \E^{\P_C}[X_t^*] \\
        &= y_{0} + t \, \eta \int_{\Rp} \!\!\! \xi \, \nu_C(d\xi) -  \frac{1}{\theta} \sum_{\kinI} \pi_k \, \ell_k(T) \, \left[1 - \exp \left( t \, \eta \int_{\R_+} \!\!\! \left( 1 - (1+\eta) \frac{v_C(\xi)}{v_k(\xi)} \right) v_C(\xi) \, d\xi  \right)\right] \,. \qedhere
    \end{align*}
\end{proof}

\begin{remark}
    Under the insurer's optimal measure $\Q^*$, the counterparty's wealth remains constant on average, i.e., for $t \in [0,T]$, $\E^{\Q^*}[Y_t^\eta] = y_{0}$. This follows by writing the equation for $Y^\eta_t$ as 
    \begin{equation*}
        Y^{\eta}_t = y_{0}  - \int_0^t \!\! \int_{\Rp} \!\!\! \alpha^*(s,\xi,\bZ^*_s) \, \tilde N^{\Q^*}(d\xi,ds)  \,,
    \end{equation*}
    from which we observe that $Y^\eta$ is a $\Q^*$-martingale that starts at $y_{0}$. 
\end{remark}

From \zcref{prop:EY}, the solution to \zcref{opt:counterpartysproblem1} is the $\eta$ that maximizes the expression
\begin{equation}
\label{eqn:Y_PC}
    \E^{\P_C}[Y_T^\eta] = y_{0} + T \, \eta \int_{\Rp} \!\!\! \xi \, \nu_C(d\xi) -  \frac{1}{\theta} \sum_{\kinI} \pi_k \, \ell_k(T) \, \left[1 - \exp \left( T \, \eta \int_{\R_+} \!\! \left( 1 - (1+\eta) \frac{v_C(\xi)}{v_k(\xi)} \right) v_C(\xi) \, d\xi  \right)\right] \,.
\end{equation}
From this expression, without further assumptions, it is difficult to ascertain the existence of an optimal safety loading $\eta$. In the next subsection, we show that if there is only one reference model, then a unique optimal $\eta$ exists and is explicit. In the general setting, given a fixed set of parameters, we can search for the optimal $\eta$ numerically. An example is given at the end of \zcref{sec:data}.

\subsection{One reference model setting}
Now, we return to the case where there is only a single reference measure, $\P$ (see \zcref{sec:one-model}) and $\nu(d\xi)$ is compound Poisson with arrival rate $\lambda$ and mean severity $\mu$ defined in \eqref{eqn:lambda_asmpt}.
Assuming the insurer employs the optimal risk sharing strategy $\alpha^*$ given a fixed price $\eta\in \Rp$, the counterparty's wealth evolves as:
\begin{equation}
\label{eqn:reins_wealth}
    \tilde Y^{\eta}_t = y_{0} + \lambda (1+\eta) \int_0^t \left[ \mu - \frac{\eta}{\theta}  e^{\lambda \eta^2 (T-s)} Z^*_{s^-} \right] ds - \int_0^t \!\! \int_{\Rp} \!\!\! \left[ \xi - \frac{\eta}{\theta}  e^{\lambda \eta^2 (T-s)} Z^*_{s^-}  \right] N(d\xi,ds) \,.
\end{equation}

The mean of $\tilde Y^{\eta}_t$ for any $t \in [0,T]$ is given by \zcref{prop:EY} applied to the one-model setting:
\begin{corollary}
    For $t \in [0,T]$ and $\eta \in \Rp$,
    \begin{equation*}
        \E^{\P}\left[\tilde Y^{\eta}_t \right] = y_{0} + \eta \lambda \mu t - \frac{1}{\theta} e^{\lambda \eta^2 (T-t)} \left( e^{\lambda \eta^2 t} -  1 \right) \,.
    \end{equation*}
\end{corollary}

In this simplified setting, we can identify the optimal $\eta^*$ that solves \zcref{opt:counterpartysproblem1}. It is given in the following proposition in terms of the Lambert--$W$ function (see e.g., \citep{LambertW}).
\begin{proposition}
    The solution to \zcref{opt:counterpartysproblem1} when there is only one reference model ($k=1$), i.e., when the dynamics of $Y^\eta$ are given by \eqref{eqn:reins_wealth}, is
    \begin{equation*}
    \eta^* = \sqrt{\frac{W(\mu^2\theta^2\lambda T / 2)}{2\lambda T}}\,,
    \end{equation*}
    where $W$ is the principal branch of the Lambert--$W$ function and $\lambda$ and $\mu$ are given by \eqref{eqn:lambda_asmpt}.
\end{proposition}

\begin{proof}
    We have
\begin{equation*}
    \E^{\P}[\tilde Y_T^\eta] = y_{0} + \eta \lambda \mu T - \frac{1}{\theta} \left[e^{\lambda \eta^2 T} -  1 \right],
\end{equation*}
and observe that its second derivative w.r.t. $\eta$ is
\begin{equation*}
    - \frac{2 \lambda T}{\theta} e^{\lambda \eta^2 T} \left[ 2\lambda \eta^2 T + 1\right] < 0 \,.
\end{equation*}
Hence,  $\E^{\P}[\tilde Y_T^\eta]$ is concave in $\eta$. The first order condition in $\eta$ then  implies that the optimal $\eta$ satisfies
\begin{equation*}
    \eta e^{\eta^2 \lambda T} = \frac{\mu \theta}{2} \,,
\end{equation*}
from which it follows that
\begin{equation*}
    2 \lambda T \eta^2 e^{2 \lambda T \eta^2} = \frac{\mu^2 \theta^2 \lambda T}{2} \,.
\end{equation*}
As $\mu^2\theta^2\lambda T / 2>0$, this equation has the unique solution on $\Rp$ (as $\eta\ge0$ to be a bonified safety loading)
\begin{equation*}
    \eta^* = \sqrt{\frac{W(\mu^2\theta^2\lambda T / 2)}{2\lambda T}} \,,
\end{equation*}
where $W$ is the principal branch of the Lambert--$W$ function.
\end{proof}

The expected wealth criterion is only one criterion the counterparty might wish to maximize. Other alternatives, such as mean-variance, may be used instead.

\section{Application to Spanish auto insurance data}
\label{sec:data}

In this section, we illustrate the optimal risk sharing strategy with model combination.
We assume an insurer has access to many data sets implying different models for the loss distribution, with no single true model. We base our example on a recent open-access insurance data set \citep{Segura-Gisbert2024-data-article,Segura-Gisbert2024-data-set}. This data set consists of 105,555 observations, giving policy-level data on annual motor insurance policies of a Spanish non-life insurer for policies that commenced in the years 2015–2018, covering claims up to the end of 2018. We use this data to estimate the model parameters using cross-validation.
We exclude policies that started in the year 2018, as the data set does not include a full year of exposure to potential claims for these policies. An analysis of the data shows that those policies do indeed have fewer claims and lower median and mean claim amounts. This leaves 69,740 observations, of which 16,259 (23\%) have positive aggregated claim amounts. 

Using cross-validation, we estimate 101 models from the data set. We estimate the counterparty's model, $\P_C$, using the full data set. For the other 100 models, denoted by $\P_k$, $ k = 1, \ldots, 100$, we sample 50\% of the data and then estimate the model parameters given that subset. 
We assume that under all models the claim arrival rate is Poisson distributed with rate $\lambda_k$ and that the severity distribution is Gamma distributed with shape parameter $m_k > 0$ and scale parameter $\phi_k > 0$ for $\kinI$. We could use any other parametric or non-parametric loss distributions, e.g., kernel density estimators, log-normal, mixture of Gammas, Weibull, or other typical loss model distributions, however, as the purpose of this section is illustrative, we opt to keep the individual models simple.

Both the arrival rate and the severity distribution are estimated from each data set by maximum likelihood. The arrival rates are estimated using the number of claims per policy. The severity distributions are fit to the average claim size per policy using the \texttt{bbmle} package in \texttt{R} \citep{bbmle}.
The parameter estimates are shown in \zcref{fig:CV_params}. The estimated parameters of the $\P_k$ models, $k=1,\ldots,100$, are in black, while the estimated parameters of $\P_C$ are in red. We observe a negative relationship between shape and scale, and no discernible relationship between the arrival rate and the shape and scale parameters. The estimated parameters for the counterparty's (full) model are arrival rate $\lambda_C = 0.52$, shape $m_C= 0.27$, and scale $\phi_C=1415$. \zcref{assump:intgr_comp} and \zcref{assump:intgr_comp_jk} hold for the parameter estimates for the models $\P_k$, $k=1,\ldots,100$ and $\P_C$.

\begin{figure}[bth]
    \centering
    \subfloat[Shape versus scale]{\includegraphics{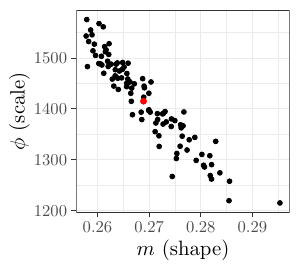}}
    \subfloat[Arrival rate versus scale]{\includegraphics{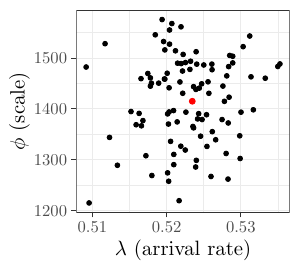}}
    \subfloat[Arrival rate versus shape]{\includegraphics{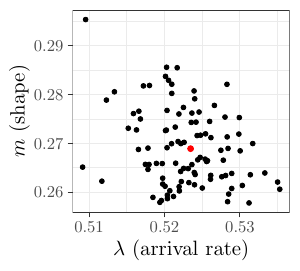}}
    \caption{Scatterplots of parameters estimated from the Spanish auto insurance data set. The parameters for the $\P_k$ models, $k=1,\ldots,100$, estimated by repeated cross-validation with 50\% of the data, are in black. The parameters for $\P_C$, estimated from the full data set, are in red.}
    \label{fig:CV_params}
\end{figure}

Using these estimated models, we implement the optimal risk sharing strategy. We simulate the jump process $\int_0^t \! \int_{\Rp} \!\!  \ln \left( (1+\eta) \frac{v_C(\xi)}{v_k(\xi)} \right)  N(d\xi,ds)$, and then compute the values of $Z^*_{k}$, $k = 1, \ldots, 100, C$, and $X^*$ using \zcref{prop:expression-Z-X}. We simulate 10,000 paths of the processes from $t=0$ to $t=T=5$. We assume that the counterparty's safety loading is $\eta=0.12$, the insurer's income is $c=5{,}550$, the weights are all equal at $\pi_k = 1/ 100$ for $k = 1,\ldots, 100$, $\pi_C=0$, and the initial value of $X$ is $x_{0} = 5{,}000$. Setting $\pi_C = 0$ means that the insurer does not use the counterparty's model.

\zcref{fig:compare-KDE} shows kernel density estimates of the distribution of the insurer's terminal wealth $X_T$ under four scenarios. 
\begin{figure}[!b]
    \centering
    \includegraphics[width=0.9\textwidth]{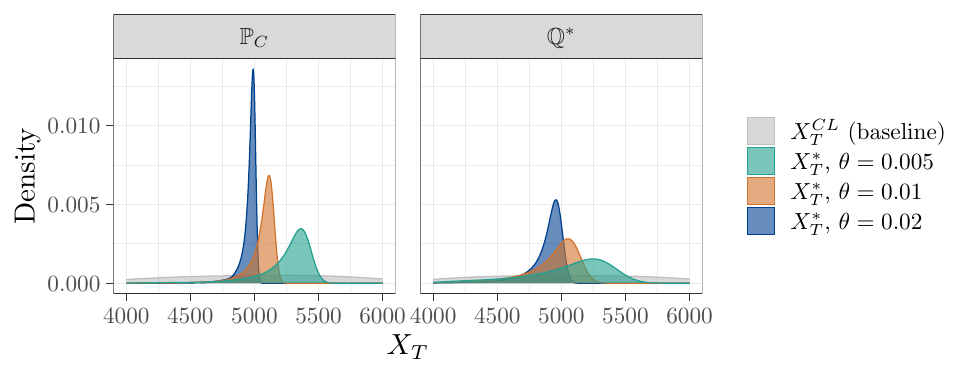}
    \caption{Kernel density estimates of the distribution of the insurer's terminal wealth $X_T$ if the insurer does not engage in risk sharing (grey) or executes the optimal risk sharing strategy with $\theta =0.02$ (blue) $\theta =0.01$ (orange), and $\theta =0.005$ (turquoise) under the probability measures $\P_C$ and $\Q^*$. The other parameters are $\eta=0.12$, $c=5{,}550$, $\pi_k = 1/ 100$ for $k = 1,\ldots, 100$, $\pi_C=0$, $x_{0} = 5{,}000$, and $T=5$.}
    \label{fig:compare-KDE}
\end{figure}
First, if the insurer does not engage in risk sharing (grey), their losses follow a Cram\'er-Lundberg process, $X^{CL}_t$. The variance of the terminal wealth $X^{CL}_T$ is high, resulting in the density curve appearing nearly flat. The other three scenarios show the distribution of the terminal wealth under the optimal strategy, $X_T^*$, for different values of the model penalization parameter $\theta$: $\theta = 0.005$ (turquoise), $\theta = 0.01$ (orange), and $\theta = 0.02$ (dark blue).  As $\theta$ increases, the insurer's optimal strategy puts more emphasis on variance reduction: the variance gets smaller under both $\P_C$ and $\Q^*$, while the mean decreases under $\P_C$. 
Estimated values of the mean and variance are given in \zcref{tab:theta-mean-var}, which shows the effect of $\theta$ on the variance and mean. For example, if $\theta = 0.005$, the variance of $X_T^*$ under $\P_C$ is about 89 thousand, a substantial reduction from the variance of $X^{CL}_T$ under $\P_C$, which is about 46.3 million. This comes at a cost of a reduction in the mean under $\P_C$, which falls from 7,914 to 5,235.

\begin{table}[!t]
    \centering
    \begin{tabular}{l@{\hskip 0.4in} r r@{\hskip 0.3in} rr}

    \toprule\toprule
    
    & \multicolumn{2}{p{0.12\textwidth}}{\centering $\P_C$} & \multicolumn{2}{p{0.12\textwidth}}{\centering$\Q^*$} 
    \\
    \midrule
    & mean & variance & mean & variance \\
    \midrule
    no risk sharing & 7,914 & 46,271,682 &  4,865 & 50,007,138 \\
    $\theta = 0.005$  & 5,235 & 89,072 & 4,865 & 790,398\\
    $\theta = 0.01$ & 5,050 & 22,268 & 4,865 & 197,600\\ 
    $\theta = 0.02$  & 4,958 & 5,567 & 4,865 & 49,400\\ 
    \bottomrule\bottomrule
    \end{tabular}
    \caption{The mean and variance of $X_T$ for different values of $\theta$ under $\P_C$ and $\Q^*$. The mean of $X_T$ and the variance of $X_T$ with no risk sharing are computed exactly, while the remaining variances are estimated numerically using 10,000 simulations. The other parameters are $\eta=0.12$, $c=5{,}550$, $\pi_k = 1/ 100$ for $k = 1,\ldots, 100$, $\pi_C=0$, $x_{0} = 5{,}000$, and $T=5$.}
    \label{tab:theta-mean-var}
\end{table}

We next illustrate the behaviour of $X^*$ and selected $Z_k^*$ over time, demonstrating the effect of model combination. As there is no true model to simulate under, we will choose to illustrate the pathwise behaviour under the pricing model, $\P_C$, but any of the 101 reference models could be used. We recall that the $Z^*$ processes are the process versions of the Radon-Nikodym derivatives between each reference model $\P_k$ and the optimal model $\Q^*$. Given the linear relationship between $X^*$ and the $Z_k^*$'s, the paths of $Z_k^*$ demonstrate the effect of the corresponding reference model $\P_k$ on the insurer's problem. \zcref{fig:XZ} shows five selected paths of $X^*$ and of two $Z^*$ variables under $\P_C$, as well as the mean (middle black line), the mean plus standard deviation of the paths above the mean (upper black line), and the mean minus standard deviation of the paths below the mean (lower black line).

For illustrative purposes, one of the selected $Z^*$ variables corresponds to a reference model which is close to the pricing model, while the other has the same arrival rate, but a different severity distribution. The severity distribution parameters for the first $Z^*$ variable illustrated are $m_{35}=0.269$ and $\phi_{35}=1444.69$, while the severity distribution parameters for the second are $m_{65}=0.282$ and $\phi_{65}=1268.71$ (compared with shape $m_C= 0.269$ and scale $\phi_C=1414.96$ under both $\P_C$ and $\Q^*$). We observe that when shape and scale parameters are closer to those of $\Q^*$, the paths of the corresponding $Z^*$ process are closer to 1, meaning these models are not distorted as much. Furthermore, the mean of $Z^*_{35}$ stays nearly constant at 1, whereas as the mean of $Z^*_{65}$ increases to 1.37 by time $T$. The negative linear relationship between $X^*$ and the $Z^*$s is also apparent, with the lower paths (red and green) of $X^*_t$ being the upper paths on $Z^*_{35}$ and $Z^*_{65}$, and the upper paths (purple and blue) of $X^*_t$ being the lower paths on $Z^*_{35}$ and $Z^*_{65}$. Therefore those reference models $\P_k$ which differ more from $\Q^*$ have a larger negative impact on the insurer's wealth under the optimal strategy, $X^*$. If one were to ignore this information, the insurer's expected wealth would be overly optimistic. In this way, the model combination framework allows for the transparent incorporation of additional information, as one can isolate the effect of each model through its corresponding $Z^*$ process.
\begin{figure}[H]
    \centering
    \hspace*{-1em}
    \subfloat[Paths of $X_t^*$ under $\P_C$]{\includegraphics[height=0.25\textheight]{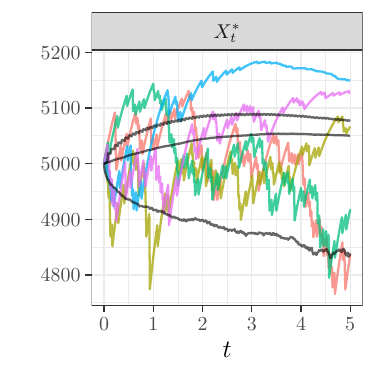}\label{fig:XZ_a}}
    \subfloat[Paths of $Z_{35,t}^*$, $Z_{65,t}^*$ under $\P_C$]{\includegraphics[height=0.25\textheight]{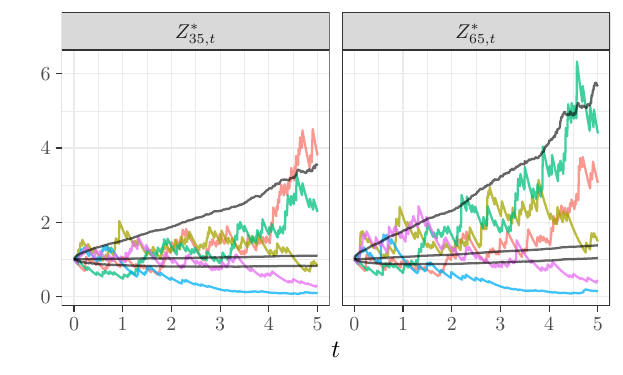}\label{fig:XZ_b}}
    \caption{Paths of $X^*_t$ and two selected $Z^*_{k,t}$ for $t \in [0,5]$ under the pricing measure $\P_C$. The parameters are $\theta = 0.01$, $\eta=0.12$, $c=5{,}550$, $\pi_k = 1/ 100$ for $k = 1,\ldots, 100$, $\pi_C=0$, $x_{0} = 5{,}000$, and $T=5$. The black lines are the mean (middle), the mean plus standard deviation of the paths above the mean (upper), and the mean minus standard deviation of the paths below the mean (lower).}
    \label{fig:XZ}
\end{figure} 

Finally, we consider the counterparty's problem, as discussed in \zcref{sec:counterparty}. Using the full 101 estimated models and the same parameters with $y_{0} = 5{,}000$ and $T=5$, we compute counterparty's terminal expected wealth \eqref{eqn:Y_PC} numerically. \zcref{fig:optim_eta_a} shows the counterparty's expected wealth under their probability measure $\P_C$ as a function of the safety loading $\eta$ for three values of $\theta$: $\theta = 0.02$ (dark blue), $\theta = 0.01$ (orange), and $\theta = 0.005$ (turquoise). The optimal safety loading, $\eta^*$, is found numerically and denoted by a dot. We see that in these cases $\E^{\P_C}[Y_T^\eta]$ is concave in $\eta$, and therefore $\eta^*$ is uniquely identified. \zcref{fig:optim_eta_b} plots the optimal safety loading  $\eta^*$, found numerically, as a function of $\theta$. We see that the optimal safety loading is increasing as the model penalization parameter, $\theta$, increases. As we have observed that increasing the model penalization parameter reduces the variance of the insurer's terminal wealth, the interpretation is that as the insurer demands less variance of their terminal wealth, the counterparty charges a higher premium.

\begin{figure}[htb!]
    \centering
    \subfloat[{$\E^{\P_C}[Y^\eta_T]$ vs.~$\eta$}]{\includegraphics{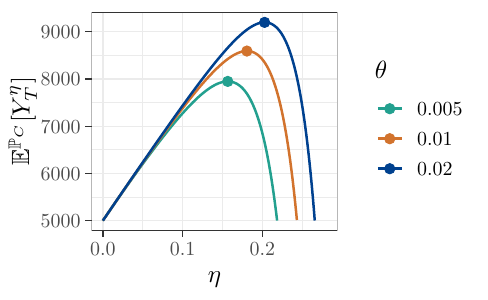}\label{fig:optim_eta_a}}
    \hspace{1em}
    \subfloat[$\eta^*$ vs.~$\theta$]{\includegraphics{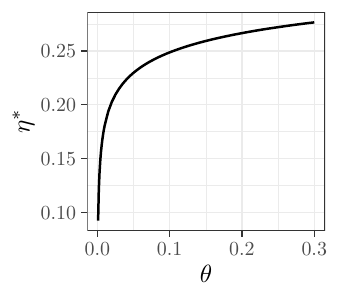}\label{fig:optim_eta_b}}
    \caption{The left panel shows $\E^{\P_C}[Y^\eta_T]$ as a function of $\eta$ for three values of $\theta$: $\theta = 0.02$ (dark blue), $\theta = 0.01$ (orange) and $\theta = 0.005$ (turquoise). The optimal $\eta$, found numerically, is denoted by the point. The right panel shows the optimal $\eta^*$ as a function of $\theta$. The other parameters are $\pi_k = 1/ 100$ for $k = 1,\ldots, 100$, $\pi_C=0$, $c=5{,}550$, $y_{0} = 5{,}000$, and $T=5$.}
    \label{fig:optim_eta}
\end{figure}

\section{Conclusion}
We study a novel risk sharing and model combination problem where an agent wishes to maximize their expected wealth when they have multiple reference models, subject to a chi-squared model ambiguity penalization. This criterion generalizes the monotone mean-variance preferences of \textcite{Maccheroni2006,Maccheroni2009} to the model combination setting, where there is no single reference model. The general problem incorporates model combination into a risk sharing problem, and could be used in various model uncertainty applications, such as risk scenarios linked to climate change or other extreme risks.

We solve this problem and find explicit solutions for the insurer's optimal risk sharing strategy, optimal decision measure, and optimal wealth process. We show that the effect of model combination is that the optimal wealth process $X^*$ depends linearly on the auxiliary processes $Z^*_k$, $\kinI$, which are the process versions of the Radon-Nikodym derivatives from each reference measure $\P_k$ to the optimal measure $\Q^*$. Using this relationship, we determine the mean and variance of $X^*$ under $\Q^*$, and find that the model penalization parameter $\theta$ penalizes the variance of the insurer's optimal wealth process. Furthermore, we show that the optimal decision measure, $\Q^*$, depends on the model used by the counterparty and the counterparty's premium, while the optimal risk sharing strategy depends on the model penalization parameter $\theta$, the counterparty's premium, and the reference models via the auxiliary processes $Z^*_k$, $\kinI$. 

There are multiple avenues for future work stemming from this problem. First, this criterion could be applied to other contexts by changing the wealth process of the agent. Our focus here has been on model combination under ambiguity with a jump process, corresponding to different scenarios for losses. The work could be expanded to consider L\'evy-It\^o processes more generally, though then one would have to consider ambiguity about joint models for both the diffusive and jump terms. It could also be extended to consider the minimally distorted stochastic process according to the chi-squared divergence given moment constraints, such as is done in \textcite{Kroell2024} with the Kullback-Leibler divergence. Future works could also consist of using different divergences, other than Kullback-Leibler and chi-squared, to penalize model ambiguity, and determine how the divergence affects the insurer's optimal strategy. Furthermore, while we took a first step in exploring the counterparty's perspective in \zcref{sec:counterparty}, this work could be further extended to a stochastic game setting, such as a leader-follower game as in \textcite{ChenShen2018}, where both agent's perspectives are fully considered. 

\section*{Acknowledgements}
We are grateful to two anonymous reviewers for helpful comments which improved the work. EK was supported in this research by an NSERC Canada Graduate Scholarship-Doctoral and an NSERC Postdoctoral Fellowship. SJ and SP would like to acknowledge support from the Natural Sciences and Engineering Research Council of Canada (grants RGPIN-2018-05705 and RGPIN-2024-04317,  and DGECR-2020-00333, RGPIN-2020-04289, RGPIN-2025-05847). SP also acknowledges support from the Canadian Statistical Sciences Institute (CANSSI).

\appendix

\section{Proof of \zcref{prop:MV_prob}}
\label{appendixA}
To prove \zcref{prop:MV_prob}, we follow the embedding approach of \textcite{ZhouLi2000}. We begin by solving the quadratic utility problem, expanding the state space to incorporate the same auxiliary processes as in the model combination problem. This gives the optimization problem:
\begin{optimization}
\label{optim:quad}
    The insurer seeks the solution to the following problem:
    \begin{align*}
        \inf_{\alpha \in \A} \frac{1}{\zeta} \E^{\Q^*} \left[\frac{( X^\alpha_T - \gamma)^2 }{\sum_{\kinI} \pi_k Z^*_{k,T}}\right]\,,
    \end{align*}
    where $\gamma\in\R$ and
    \begin{subequations}
        \begin{align}
        \label{eqn:X_Q_SDE2}
            d X_t^\alpha &= \left[ c - (1+\eta) \int_{\Rp} \!\!\! \alpha_t(\xi) \, \nu_C (d\xi)   \right] dt - \int_{\Rp} \!\!\! [\xi - \alpha_t(\xi)] \, N(d\xi,dt) \,,\\
            dZ_{k,t}^* &=  - Z_{k,t^-}^* \int_{\Rp} \!\! \left[ 1 - \frac{(1+\eta)v_C(\xi)}{v_k(\xi)}\right] \tilde N^{\P_k}(d\xi,dt)\,, \quad \text{for all } \kinI \,,
        \end{align}
    \end{subequations}
    with $X^\alpha_0 = x_{0}$, $Z_{k,0}^*=1$ for all $\kinI$.
\end{optimization}

We first solve \zcref{optim:quad} for $\gamma\in\R$ fixed and find the optimal strategy $\hat \alpha^\gamma$ and the corresponding optimal wealth process under this strategy, $\hat X^\gamma := X^{\hat \alpha^\gamma}$. Then we apply the result of \textcite[Theorem 3.1]{ZhouLi2000},  which states that $\hat \alpha^{\gamma^*}$ solves \zcref{optim:MV}, where $\gamma^*$ is the solution to the equation
$\gamma = 1/\bar\theta + \E^{\bar\P}[\hat X_T^\gamma]$ in $\gamma$. This gives us the optimal mean-variance strategy.

\begin{proposition}
\label{prop:quad_soln}
    The candidate optimal control for \zcref{optim:quad} is
    \begin{equation*}
        \alpha^\gamma(t,\hat X^\gamma_t,\bZ^*_t) = \xi + \left( \hat X^\gamma_t - A^\gamma(t) \right) \left( \frac{\sum_{\kinI} \pi_k \ell_k(T-t) Z^*_{k,t}\left[ \frac{(1+\eta) v_C(\xi)}{v_k(\xi)}-1\right]}{\sum_{\kinI} \pi_k \ell_k(T-t) Z^*_{k,t}} \right) \,,
    \end{equation*}
    where
    \begin{equation*}
        A^\gamma(t) = \gamma + \left[ (1+\eta) \int_{\Rp} \!\!\! \xi v_C(\xi)\,d\xi - c\right](T-t) \,,
    \end{equation*}
    The candidate value function is
    \begin{equation*}
        V^\gamma(t,x,\bz) = \frac{(x-A^\gamma(t))^2}{\sum_{\kinI} \pi_k \ell_k(T-t) z_{k}} \,,
    \end{equation*}
    where $\ell_k(t)$ is defined in \eqref{eqn:ell_k}.
\end{proposition}

\begin{proof}
    Define the value function
    \begin{equation*}
        V^\gamma(t,x,\bz) = \inf_{\alpha^\gamma \in \A} \E^{\Q^*}_{t,x,\bz}  \left[\frac{( X^\alpha_T - \gamma)^2 }{\sum_{\kinI} \pi_k Z^*_{k,T}}\right],
    \end{equation*}
    which satisfies the HJB equation
    \begin{align*}
     0 & = \, \partial_t V^\gamma(t,x,\bz)  + \sum_{\kinI} z_{k} \int_{\Rp} \!\! \left[ v_k(\xi) - (1+\eta)v_C(\xi)\right]  \,\partial_{z_k} V^\gamma(t,x,\bz) \\
    & \quad + \inf_{a \in \mfS} \Bigg\{ \! \left[ c - (1+\eta) \int_{\Rp} \!\!\! a(t,\xi,x,\bz) \, \nu_C (d\xi) \right] \partial_x V^\gamma(t,x,\bz)  \\
    & \hspace{5.5em}  + \int_{\Rp} \!\! \Big[ V^\gamma \! \left(t,x - [\xi -  a(t,\xi,x,z)],z_1 \frac{(1+\eta)v_C(\xi)}{v_1(\xi)},\ldots,z_n \frac{(1+\eta)v_C(\xi)}{v_n(\xi)},z_C (1+\eta) \right) \\
    &\hspace{8.5em}  - V^\gamma(t,x,\bz) \Big] (1+\eta)v_C(\xi) d\xi \Bigg\}
\end{align*}
with terminal condition $V^\gamma(T,x,\bz) = (x - \gamma)^2/(\sum_{\kinI} \pi_k z_k)$.
We use the Ansatz 
\begin{equation*}
    V^\gamma(t,x,\bz) = \frac{(x-A^\gamma(t))^2}{\sum_{\kinI} \pi_k \ell_k(T-t) z_k} \,,
\end{equation*}
where
 $\ell_k(t) = \exp \left( t \int_{\Rp} \!\! \left[1- (1+\eta) \frac{v_C(\xi)}{v_k(\xi)} \right]^2\!\! \nu_k(d\xi)\right)$ is the same as in the solution to the model combination, \zcref{prop:optim_contr}, and $A^\gamma(t)$ is an unknown function with terminal condition $A^\gamma(T)=\gamma$. Substituting the Ansatz into the HJB equation and simplifying gives
\begin{align*}
    0 &= -2 (x-A^\gamma(t)){A^\gamma}'(t) + 2c (x-A(t)) \\
    & \quad + \frac{(x-A^\gamma(t))^2}{\sum_{\kinI} \pi_k \ell_k(T-t) z_k} \sum_{\kinI} \pi_k \ell_k(T-t) z_k \int_{\Rp} \left[(1+\eta)^2 \frac{v_C^2(\xi)}{v_k(\xi)} - 2 (1+\eta) v_C(\xi) \right] d\xi \\
    &\quad + \inf_{a \in \mfS} \Bigg\{ \! -2 (x-A^\gamma(t)) (1+\eta) \int_{\Rp} \!\!\! a(t,\xi,x,\bz) \,v_C (\xi)d\xi \\
    & \hspace{5em} + \int_{\Rp}\!\! [x - (\xi - a(t,\xi,x,\bz)) - A^\gamma(t)]^2 \frac{\sum_{\kinI} \pi_k \ell_k(T-t) z_k }{\sum_{\kinI} \pi_k \ell_k(T-t) \frac{z_k}{v_k(\xi)} } d\xi \Bigg\} \,.
\end{align*}
 Solving the infimum problem yields
\begin{equation*}
    \hat a^\gamma(t,\xi,x,\bz) = \xi + \left(x -A^\gamma(t) \right) \frac{\sum_{\kinI} \pi_k \ell_k(T-t) z_k \left(\frac{(1+\eta) v_c(\xi)}{v_k(\xi)} -1 \right)}{\sum_{\kinI} \pi_k \ell_k(T-t) z_k }\,.
\end{equation*}
We substitute $\hat a$ into the HJB equation. Tedious calculations show that it simplifies to 
\begin{equation*}
    0 = 2 (x-A^\gamma(t)) \left[-{A^\gamma}'(t) + c - (1+\eta) \int_{\Rp}\!\! \xi \, v_C(\xi) \,d\xi \right]\,.
\end{equation*}
Given the terminal condition $A^\gamma(T)=\gamma$, we have $A^\gamma(t) = \gamma + \left[ (1+\eta) \int_{\Rp} \!\!\! \xi v_C(\xi)\,d\xi - c\right](T-t)$.
\end{proof}

\begin{lemma}
\label{prop:XA}
For every $t \in[0,T]$, we have
\begin{equation*}
    \frac{\hat X_t^\gamma - A^\gamma(t)}{\sum_{\kinI} \pi_k \ell_k(T-t) Z^*_{k,t}} = \frac{x_0 - A^\gamma(0)}{\sum_{\kinI} \pi_k \ell_k(T)} \,.
\end{equation*}
\end{lemma}
\begin{proof}
    Substituting $\hat\alpha^\gamma$ into the dynamics of $X_t$, \eqref{eqn:X_Q_SDE2}, gives
    \begin{align*}
        d \hat X_t^\gamma &= \left[ c - (1+\eta) \int_{\Rp} \!\!\! \xi \, \nu_C (d\xi) - \int_{\Rp}\!\! \left(\hat X_t^\gamma -A^\gamma(t) \right) \frac{\sum_{\kinI} \pi_k \ell_k(T-t) z_k \left(\frac{(1+\eta) v_c(\xi)}{v_k(\xi)} -1 \right)}{\sum_{\kinI} \pi_k \ell_k(T-t) z_k } (1+\eta) v_C(\xi) d\xi \right] dt \\
        &\quad + \int_{\Rp} \!\!\! \left(\hat X_t^\gamma -A^\gamma(t) \right) \frac{\sum_{\kinI} \pi_k \ell_k(T-t) z_k \left(\frac{(1+\eta) v_c(\xi)}{v_k(\xi)} -1 \right)}{\sum_{\kinI} \pi_k \ell_k(T-t) z_k } \, N(d\xi,dt)
    \end{align*}
    from which it is clear that
    \begin{align*}
        d (\hat X_t^\gamma - A^\gamma(t)) &= \left(\hat X_t^\gamma -A^\gamma(t) \right) \int_{\Rp} \!\!\! \frac{\sum_{\kinI} \pi_k \ell_k(T-t) z_k \left(\frac{(1+\eta) v_c(\xi)}{v_k(\xi)} -1 \right)}{\sum_{\kinI} \pi_k \ell_k(T-t) z_k } \, \tilde N^{\Q^*}(d\xi,dt) \,.
    \end{align*}
    Furthermore, by It\^o's Lemma,
    \begin{align*}
        d \left(\sum_{\kinI} \pi_k \ell_k(T-t) Z^*_{k,t} \right)&= \sum_{\kinI} \pi_k \ell_k(T-t) Z^*_{k,t} \int_{\Rp} \!\!\! \frac{\sum_{\kinI} \pi_k \ell_k(T-t) z_k \left(\frac{(1+\eta) v_c(\xi)}{v_k(\xi)} -1 \right)}{\sum_{\kinI} \pi_k \ell_k(T-t) z_k } \, \tilde N^{\Q^*}(d\xi,dt) \,.
    \end{align*}
    Given the similar forms of the differentials for $\hat X_t^\gamma - A^\gamma(t)$ and $\sum_{\kinI} \pi_k \ell_k(T-t) Z^*_{k,t}$, another application of It\^o's Lemma gives
    \begin{equation*}
        d\left( \frac{\hat X_t^\gamma - A^\gamma(t)}{\sum_{\kinI} \pi_k \ell_k(T-t) Z^*_{k,t}}\right) = 0 \,,
    \end{equation*}
    from which the conclusion follows. 
\end{proof}

\begin{proposition}
    The candidate control and value function given in \zcref{prop:quad_soln} are optimal for \zcref{optim:quad}.
\end{proposition}

We will not give a proof of the verification theorem here as the arguments are standard. To see why, observe that in light of \zcref{prop:XA}, we can re-write the candidate control and value functions in feedback form as
    \begin{align*}
        \alpha^{\gamma}(t,\hat X_t,\bZ^*_t) = \xi + \left( \frac{x_0 - A^{\gamma}\!(0)}{\sum_{\kinI} \pi_k \ell_k(T)} \right) \left( {\sum_{\kinI} \pi_k \ell_k(T-t) Z^*_{k,t}\left[ \frac{(1+\eta) v_C(\xi)}{v_k(\xi)}-1\right]}\right)\,, 
    \end{align*}
    and
    \begin{equation*}
    V^\gamma(t,\hat X_t,\bZ_t^*)  = \frac{x_0 - A^\gamma(0)}{\sum_{\kinI} \pi_k \ell_k(T)} (\hat X_t - A^\gamma(t)) \,.
\end{equation*}
Then arguments similar to those in the proof of \zcref{prop:alpha_adm} show that $\alpha^\gamma$ is an admissible strategy. The candidate value function is linear in $\hat X_t$, and so standard arguments show that it is the optimal value function. 

We are now in a position to give the proof of \zcref{prop:MV_prob}:

\begin{proof}
By the definition of $\bar \P$ in \eqref{eqn:Pbar_def} and \zcref{prop:P_bar_unique}, we have
\begin{equation*}
    \frac{d\bar\P}{d\Q^*} = \frac{1}{\zeta}\frac{1}{\sum_{\kinI} \pi_k  \frac{d\Q^*}{d\P_k}} =  \frac{1}{\zeta} \frac{1}{\sum_{\kinI} \pi_k  Z^*_{k,T}}\,.
\end{equation*}
    Then from \zcref{prop:XA}, it follows that
\begin{align*}
    \E^{\bar \P} \left[ X_T^*\right] = \frac{1}{\zeta} \E^{\Q^*} \left[ \frac{X_T^*}{\sum_{\kinI} \pi_k Z^*_{k,T}}\right] &= \frac{1}{\zeta} \frac{x_0 - A^\gamma(0)}{\sum_{\kinI} \pi_k \ell_k(T)} + \frac{1}{\zeta} \E^{\Q^*} \left[ \frac{A^\gamma(T)}{\sum_{\kinI} \pi_k Z^*_{k,T}}\right] \\
    &= \frac{1}{\zeta} \frac{x_0 - A^\gamma(0)}{\sum_{\kinI} \pi_k \ell_k(T)} + A^\gamma(T) \,.
\end{align*}

We compute $\gamma^*$ using the condition from \textcite{ZhouLi2000}. We have that $\gamma^*$ solves
\begin{align*}
    \gamma = \frac{1}{\bar \theta} + \E^{\bar \P} \left[ X_T^*\right] &= \frac{1}{\zeta \theta} +  \frac{1}{\zeta} \frac{x_0 -  A^\gamma(0)}{\sum_{\kinI} \pi_k \ell_k(T)} + A^\gamma(T) \\
     &= \frac{1}{\zeta \theta} +  \frac{1}{\zeta} \frac{x_0 - \gamma - \left[ 1+\eta) \int_{\Rp} \xi v_C(\xi)d\xi - c\right]T}{\sum_{\kinI} \pi_k \ell_k(T)} + \gamma \,,
\end{align*}
which gives
\begin{equation*}
    \gamma^* = \frac{\sum_{\kinI} \pi_k \ell_k(T)}{\theta} + x_0 - \left[ (1+\eta) \int_{\Rp} \xi v_C(\xi)d\xi - c\right]T \,.
\end{equation*}
Then we re-write the optimal risk sharing strategy from \zcref{prop:quad_soln} as follows, where $\hat X_t^*$ is the solution to \eqref{eqn:X_Q_SDE2} with $\alpha = \hat \alpha^{\gamma^*}$, $\hat X_0^*=x_0$:
\begin{align*}
        \alpha^{\gamma^*}(t,\hat X_t^*,\bZ^*_t) &= \xi + \left( \frac{\hat X_t^* - A^{\gamma^*}\!(t)}{\sum_{\kinI} \pi_k \ell_k(T-t) Z^*_{k,t}} \right) \left( {\sum_{\kinI} \pi_k \ell_k(T-t) Z^*_{k,t}\left[ \frac{(1+\eta) v_C(\xi)}{v_k(\xi)}-1\right]}\right) \\
        &= \xi + \left( \frac{x_0 - A^{\gamma^*}\!(0)}{\sum_{\kinI} \pi_k \ell_k(T)} \right) \left( {\sum_{\kinI} \pi_k \ell_k(T-t) Z^*_{k,t}\left[ \frac{(1+\eta) v_C(\xi)}{v_k(\xi)}-1\right]}\right) \\
        &= \xi + \left( \frac{x_0 - \gamma^* - \left[ (1+\eta) \int_{\Rp} \!\!\! \xi v_C(\xi)d\xi - c\right]T}{\sum_{\kinI} \pi_k \ell_k(T)} \right) \left( {\sum_{\kinI} \pi_k \ell_k(T-t) Z^*_{k,t}\left[ \frac{(1+\eta) v_C(\xi)}{v_k(\xi)}-1\right]}\right) \\
        &= \xi - \frac{1}{\theta} {\sum_{\kinI} \pi_k \ell_k(T-t) Z^*_{k,t}\left[ \frac{(1+\eta) v_C(\xi)}{v_k(\xi)}-1\right]} \,,
    \end{align*}
    which shows that the strategy is identical to the optimal model combination strategy given in \zcref{prop:optim_contr}.
\end{proof}

\section{Proof of \zcref{prop:P-bar-compensator}}
\label{appendixB}

Choose $\P$ equivalent to $\P_k$, where the $\P$-compensator of $N$ is denoted $v$, such that (with a slight abuse of notation) $v(d\xi,dt) = v(\xi)\,d\xi\,dt$ and $\int_{\Rp} \! v(\xi)\,d\xi<\infty$. Define the composite probability measure $\bar \P$ as
\begin{equation*}
    \frac{d\bar\P}{d\P} = \frac1\zeta \left( \sum_{k\in\cI} \pi_k \frac{d\P\,}{d\P_k}\right)^{-1} \!, \quad \text{where } \zeta := \E^{\P} \left[ \left( \sum_{k\in\cI} \pi_k \frac{d\P\,}{d\P_k}\right)^{-1}\right] \,.
\end{equation*}
and the conditional Radon-Nikodym derivatives
\begin{align}
    Z_t^k &:= \E^{\P}\left[\left.\frac{d\P_k}{d{\P}} \right|\F_t\right], \qquad \forall \kinI, \nonumber
\\
    \bar{Z}_t &:= \E^{\P}\left[\left.\frac{d\bar\P}{d{\P}}\right |\F_t\right]
    =  \E^{\P}\left[\left.\left(\zeta\,
    \sum_{k\in\cI} \pi_k \frac{d\P\,}{d\P_k}\right)^{-1}\right|\F_t\right]
    =
    \zeta^{-1}\; \E^{\P}\left[\left.\left(
    \sum_{k\in\cI} \pi_k (Z_T^k)^{-1}\right)^{-1}\right|\F_t\right].
    \label{eqn:bar-Z}
\end{align}

We aim to find the $\bar\P$-compensator of $N$. Similar to \eqref{eqn:Z_SDE} we have that
\begin{align*}
    dZ_t^k = -Z_{t^-}^k \int_{\Rp}\!\! \left(1-\frac{v_k(\xi)}{ v(\xi)}\right) \left(N(d\xi,dt)- v(d\xi,dt)\right).
\end{align*}
From \eqref{eqn:bar-Z} and the above, we have that $\bar Z_t$ is Markov in $(Z_t^k)_{\kinI}$, hence we may write $\bar Z_t = f(t,Z_t^1,\dots,Z_t^n,Z_t^C)$ for a function $f$ that satisfies the PIDE
\begin{equation}
\left\{
\hspace{-1em}
\begin{split}
    \partial_t f(t,z_1,\dots,z_n,z_C) \hspace{22em}\\
    \hspace{3em}+ \int_{\Rp}\!\! v(\xi)\left(f\left(t,\tfrac{v_1(\xi)}{v(\xi)}\,z_1,\dots,\tfrac{v_n(\xi)}{v(\xi)}z_n,\tfrac{v_C(\xi)}{v(\xi)}z_C\right)
    -f(t,z_1,\dots,z_n,z_C)\right) d\xi
    &=0,
    \\
    f(T,z_1,\dots,z_n,z_C) &= 
    \left(\zeta\,\sum_{k\in\cI}  \frac{\pi_k}{z_k}\right)^{-1}.
\end{split}
\right.
\label{eqn:PIDE-f}
\end{equation}
Under the integrability assumptions in \eqref{eqn:nu_int_assumption} and on $v(\xi)$, and the equivalence of the probability measures, standard Picard iteration arguments show that a unique solution of this equation exists.
Given a solution to \eqref{eqn:PIDE-f}, from It\^o's lemma and that $\bar{Z}$ is a $\P$-martingale, we then have that
\begin{align*}
    d\bar{Z}_t &= \int_{\Rp}\!\! 
    \left(f\left(t,\tfrac{v_1(\xi)}{v(\xi)}\,Z^1_{t^-},\dots,\tfrac{v_n(\xi)}{v(\xi)}Z^n_{t^-},\tfrac{v_C(\xi)}{v(\xi)}Z^C_{t^-}\right)
    -f(t,Z^1_{t^-},\dots,Z^n_{t^-},Z^C_{t^-})\right)\,\left(N(d\xi,dt)
    - v(d\xi,dt)\right)
    \\
    &= -\bar{Z}_{t^-} \int_{\Rp}\!\! 
    \left(1-\frac{f\left(t,\tfrac{v_1(\xi)}{v(\xi)}\,Z^1_{t^-},\dots,\tfrac{v_n(\xi)}{v(\xi)}Z^n_{t^-},\tfrac{v_C(\xi)}{v(\xi)}Z^C_{t^-}\right)}{f(t,Z^1_{t^-},\dots,Z^n_{t^-},Z^C_{t^-})}
    \right)\,\left(N(d\xi,dt)
    - v(d\xi,dt)\right)\,.
\end{align*}
Therefore, the $\bar{\P}$-compensator of $N$ is
\begin{equation*}
    \bar{v}(t,\xi) = \frac{f\left(t,\tfrac{v_1(\xi)}{v(\xi)}\,Z^1_{t^-},\dots,\tfrac{v_n(\xi)}{v(\xi)}Z^n_{t^-},\tfrac{v_C(\xi)}{v(\xi)}Z^C_{t^-}\right)}{f(t,Z^1_{t^-},\dots,Z^n_{t^-},Z^C_{t^-})}\;v(\xi)\;.
\end{equation*}

We next, show that the $f(t,z_1,\dots,z_n,z_C)$ is homogenous of order one in $z$'s. First, the terminal condition gives
\begin{equation*}
    f(T,c\,z_1,\dots,c\,z_n,c\,z_C) = 
    \left(\zeta\,\sum_{k\in\cI} \pi_k (c\,z_k)^{-1}\right)^{-1}
    = c f(T,z_1,\dots,z_n,z_C)\,.
\end{equation*}
Next, define the integral operator
\begin{equation*}
    (\mathcal L f)(t,z_1,\dots,z_n,z_C) := \int_{\Rp}\!\! v(\xi)\left(f\left(t,\tfrac{v_1(\xi)}{v(\xi)}\,z_1,\dots,\tfrac{v_n(\xi)}{v(\xi)}z_n,\tfrac{v_C(\xi)}{v(\xi)}z_C\right)
    -f(t,z_1,\dots,z_n,z_C)\right) d\xi \,.
\end{equation*}
Suppose that $h$ is homogenous of order one in $z$, then we show that $(\mathcal L h)$ is as well. To this end,
\begin{align*}
    (\mathcal L h)(t,c\,\bz) 
    &=
    \int_{\Rp}\!\! v(\xi)\left(h\left(t,c\,\tfrac{v_1(\xi)}{v(\xi)}\,z_1,\dots,c\,\tfrac{v_n(\xi)}{v(\xi)}z_n,c\,\tfrac{v_C(\xi)}{v(\xi)}z_C\right)
    -h(t,c\,z_1,\dots,c\,z_n,c\,z_C)\right) d\xi
    \\
    &=
    c\int_{\Rp}\!\! v(\xi)\left(h\left(t,\tfrac{v_1(\xi)}{v(\xi)}\,z_1,\dots,\tfrac{v_n(\xi)}{v(\xi)}z_n,\tfrac{v_C(\xi)}{v(\xi)}z_C\right)
        -h(t,z_1,\dots,z_n,z_C)\right) d\xi    
    \\
    &= c\,(\mathcal L h)(t,\bz) 
\end{align*}

So next, define $h_c(t,\bz) := \frac{1}{c} f(t,c\bz)$. Then $h_c(T,\bz) = 
    \left(\zeta\, \sum_{k\in\cI} \pi_k (z_k)^{-1}\right)^{-1}$.
Next, $\partial_t h_c(t,\bz) = \tfrac{1}{c} \partial_t f(t,c\,\bz)$.
Further, as $\mathcal L$ is a linear operator, $(\mathcal L h_c)(t,\bz) 
    = \tfrac1c (\mathcal L f)(t,c\,\bz)$. 
Therefore,
\begin{equation*}
    \partial_t h_c(t,z) + (\mathcal L h_c)(t,\bz) = \tfrac1c \partial_t f(t,c\bz) + \tfrac 1c (\mathcal L f)(t,c\,\bz) =0\,.
\end{equation*}
Therefore, $h_c$ solves the same PIDE as $f$ and with the same terminal condition. As the solution to the PIDE is unique, then $h_c$ and $f$ coincide. Therefore, we must have $\tfrac1c f(t,c\bz) = f(t,\bz)$, implying $f(t,c\bz) = c\,f(t,\bz)$.

Back to $\bar{v}$, we have
\begin{align*}
    \bar{v}(t,\xi) 
    &= \frac{f\left(t,\tfrac{v_1(\xi)}{v(\xi)}\,Z^1_{t^-},\dots,\tfrac{v_n(\xi)}{v(\xi)}Z^n_{t^-},\tfrac{v_C(\xi)}{v(\xi)}Z^C_{t^-}\right)}{f(t,Z^1_{t^-},\dots,Z^n_{t^-},Z^C_{t^-})}\;v(\xi) = \frac{\frac{1}{v(\xi)}f\left(t,v_1(\xi)Z^1_{t^-},\dots,v_n(\xi)Z^n_{t^-},v_C(\xi)Z^C_{t^-}\right)}{f(t,Z^1_{t^-},\dots,Z^n_{t^-},Z^C_{t^-})}\;v(\xi)
    \\
    &= \frac{f\left(t,v_1(\xi)Z^1_{t^-},\dots,v_n(\xi)Z^n_{t^-},v_C(\xi)Z^C_{t^-}\right)}{f(t,Z^1_{t^-},\dots,Z^n_{t^-},Z^C_{t^-})} \,.
\end{align*}
Further, note that if we define $g(t,z_1,\ldots,z_n,z_C):=e^{\int_{\Rp}\!\! v(\xi)\,d\xi (T-t)}\,f(t,z_1,\dots,z_n,z_C)$, then $g$ inherits the homogeneity property and $g$ satisfies the PIDE 
\begin{align*}
    0&=\partial_t g(t,z_1,\dots,z_n,z_C)
    + \int_{\Rp}\!\! v(\xi) \,g\left(t,\tfrac{v_1(\xi)}{v(\xi)}\,z_1,\dots,\tfrac{v_n(\xi)}{v(\xi)}z_n,\tfrac{v_C(\xi)}{v(\xi)}z_C\right)
     d\xi \\
    &= \partial_t g(t,z_1,\dots,z_n,z_C)
    + \int_{\Rp}\!\! g\left(t, v_1(\xi) z_1,\dots,v_n(\xi)z_n,v_C(\xi)z_C\right)
     d\xi  
\end{align*}
with the same terminal condition as $f$, which concludes the proof.
\qed

\singlespacing
\printbibliography

\end{document}